\documentclass[letterpaper, 10 pt, conference]{ieeeconf}  
\IEEEoverridecommandlockouts                              
\title{\LARGE \bf
Error-In-Variables Methods for Efficient System Identification with Finite-Sample Guarantees}

\author{Yuyang Zhang$^{1,2}$, Xinhe Zhang$^{1}$, Jia Liu$^{1}$ and Na Li$^{1}$
\thanks{This work is supported by NSF AI institute 2112085, NSF
ECCS 2328241, NIH R01LM014465. Yuyang Zhang is supported by the Kempner Graduate Fellowship.}
\thanks{$^{1}$Y. Zhang, X. Zhang, J. Liu and N. Li are with SEAS, Harvard University, USA.
{\tt\small \{yuyangzhang@g, xinhezhang@g, jia\_liu@seas nali@seas\}.harvard.edu.}
}%
\thanks{$^{2}$Yuyang Zhang is also with the Kempner Institute, Harvard University.
}%
}

\makeatletter
\let\NAT@parse\undefined
\makeatother

\usepackage{edgeStd}

\newcommand{\ls}{\text{LS}}
\newcommand{\iv}{\text{IV}}
\newcommand{\bc}{\text{BC}}

\begin{document}
\maketitle
\thispagestyle{empty}
\pagestyle{empty}

\begin{abstract}
    This paper addresses the problem of learning linear dynamical systems from noisy observations. In this setting, existing algorithms either yield biased parameter estimates or have large sample complexities. We resolve these issues by adapting the instrumental variable method and the bias compensation method, originally proposed for error-in-variables models, to our setting. We provide refined non-asymptotic analysis for both methods. Under mild conditions, our algorithms achieve superior sample complexities that match the best-known sample complexity for learning a fully observable system without observation noise. 
\end{abstract}

\section{Introduction}

Dynamical systems play a central role in numerous real-world applications, spanning neuroscience \cite{app_neuro1,app_neuro2}, climate \cite{app_climate1}, power grids \cite{app_power1}, robotics \cite{app_robot1,app_robot2}, and beyond. 
In many of these domains, the underlying system models are unknown, making it challenging to perform critical downstream tasks including state prediction \cite{app_pred}, false detection \cite{app_false}, and feedback control \cite{app_feedback}. In response, researchers have been developing data-driven algorithms for learning the system models, which are often referred to as system identification (SYSID) algorithms \cite{sysid_00,ljung1998system}.

In the context of linear systems, recent research primarily focuses on two classes of systems, i.e. fully observable and partially observable systems. In the fully observable setting, the entire state vectors are observed without observation noises. To learn the linear system models, it suffices to perform least-squares on the state and input trajectories. Recent studies \cite{sysid_ls1,sysid_ls2,sysid_ls3} have established non-asymptotic error bounds of order $\tilde{\calO}(\sqrt{(n+m)/T})$ for an $n$-dimensional systems with $m$-dimensional inputs and $T$ data samples, which are optimal up to constants and logarithmic factors. In partially observable systems, the states are measured through a linear observer with observation noises. In this setting, the most widely used algorithms are Ho-Kalman-based algorithms \cite{sysid_1,sysid_markov1}. These methods first estimate the input-to-output mapping of the system,
which is composed of the so-called Markov parameters. The dynamical model, or the system matrices, is subsequently recovered from the Markov parameters. They achieve non-asymptotic error upper bounds of order $\tilde{\calO}(\sqrt{n^{x}(m+n+d)/T})$ for systems with $n$-dimensional states, $m$-dimensional input, $d$-dimensional observations and $T$ data samples. Here $x$ is some positive integer depending on the specific algorithm.

For an intermediate class of systems where the entire state vectors are observed but subject to observation noises, the aforementioned algorithms do not perform ideally. For detailed illustrations, consider the following linear system with $n$-dimensional states and $m$-dimensional inputs
\begin{equation}\begin{split}\label{eq:sys0}
    x_{t+1} = Ax_t + Bu_t + w_t, \quad \htx_t = x_t + \eta_t.
\end{split}\end{equation}
Here $x_{t}\in\bbR^n, u_{t}\in\bbR^m, w_{t}\in\bbR^n$ denote the state, input and process noise, respectively. $\htx_t\in\bbR^n$ and $\eta_t\in\bbR^n$ denote the state observation and observation noise.
On one hand, directly applying least-squares on $(\htx_t,u_t)$ and $\htx_{t+1}$, as in the fully observable case, leads to biased estimations of matrices $A$ and $B$. 
Because both $\htx_t$ and $\htx_{t+1}$ are contaminated by observation noises, the above problem becomes an error-in-variable (EIV) problem \cite{sysid_eiv,sysid_eiv2} where the least-squares estimator is known to be biased and does not converge to the true parameter regardless of the number of samples.
On the other hand, we can treat the system as a partially observable system with an identity observer matrix $C=I_n$ and apply Ho-Kalman-based algorithms. However, this approach leads to a suboptimal estimation error.
This is because under the standard controllability assumption, Ho-Kalman-based algorithms need to estimate the Markov parameters up to order $n$, i.e. $\{A^iB\}_{i=0}^n$, from data. Although the estimation is unbiased, the algorithms need to learn approximately $mn/(m+n)$ times more parameters than necessary, thereby inflating the sample complexity by at least a factor of $mn/(m+n)$. 


In this paper, we propose sample-efficient algorithms to learn the system in \Cref{eq:sys0}. Specifically, we adapt two classical methods in the EIV problems, i.e. the instrumental variable (IV) method and the bias-compensation (BC) method \cite{sysid_eiv}, to our setting. We provide finite-sample analysis of the proposed algorithms, establishing an error upper bound of order $\tilde{\calO}(\sqrt{(m+n)/T})$ on the estimated system matrices under mild assumptions. Here $\tilde{\calO}(\cdot)$ hides constants and logarithmic factors. Notably, this rate matches the error upper bound for learning a fully observable linear system without any observation noises \cite{sysid_ls1}. Namely, our algorithms effectively mitigate the influence of the observation noises while only sacrificing some noise-related constants and logarithmic factors. 
Achieving this tight upper bound requires careful non-asymptotic analysis, which is fundamentally different from the asymptotic analysis in existing EIV literature \cite{sysid_eiv,sysid_eiv2}. 


\textit{Notations.} We use $\lesssim$ and $\gtrsim$ to hide absolute constants. We use $\tilde{\calO}(\cdot)$ to hide all problem-related constants, absolute constants and logarithmic factors. For any positive integer $a$, let $[a]$ denote set $\{1, 2, \cdots, a\}$. 
For matrix $M\in\bbR^{m\times n}$, we let $\sigma_{1}(M) \geq \sigma_2(M)\geq \cdots \geq \sigma_{\max\{m,n\}}(M)\geq0$ denote its singular values and let $\norm{M}$ denote its operator norm.

\section{Problem Setup \& Preliminary}\label{sc:prelim}
\subsection{System Model \& Learning Objective}
We study the linear dynamical system in \Cref{eq:sys0} with $n$-dimensional states and $m$-dimensional inputs.
For simplicity, we assume $x_{0} = 0$. We also assume $u_{t}\overset{\text{i.i.d.}}{\sim}\calN(0,\Sigma_u)$, $w_{t}\overset{\text{i.i.d.}}{\sim}\calN(0,\Sigma_w)$, $\eta_{t}\overset{\text{i.i.d.}}{\sim}\calN(0,\Sigma_{\eta})$\footnote{Our analysis can be easily adapted to inputs and noises with subGaussian distributions.}. 
We denote the above system, along with noises in the inputs and observation noises, by $\calM=(A,B,\Sigma_w,\Sigma_u,\Sigma_{\eta})$.

With the inputs $\{u_t\}_{t=0}^{T-1}$ and state measurements $\{\htx_t\}_{t=0}^{T}$, we aim to learn the system matrices $A, B$ up to $\epsilon$-accuracy:
\begin{equation*}\begin{split}
    \max\left\{\norm{A-\htA}, \norm{B-\htB}\right\} \leq \epsilon.
\end{split}\end{equation*}

\subsection{Necessary Assumptions}

Throughout the paper, we consider a target system $\calM=(A,B,\Sigma_w,\Sigma_u,\Sigma_\eta)$
satisfying the following two assumptions standard in related literature \cite{sysid_1,proof_1}.
\begin{assumption}\label{assmp3:stable}
    There exist positive constants $\rho_A < 1$ and $\psi_A\geq 1$ such that
    \begin{equation*}\begin{split}
        \pushQED{\qed} 
        \norm{A^t} \leq \psi_A\rho_A^{t-1}, \quad \forall t\geq 0.\qedhere
        \popQED
    \end{split}\end{equation*}   
\end{assumption}
\noindent Intuitively, it assumes the $A$ matrix is ``stable'' so that the system states do not blow up over a long horizon. 
    
\begin{assumption}\label{assmp2:controllable}
    $(A,B)$ is controllable. Or equivalently, the controllability matrix $R = \begin{bmatrix}
        B & AB & \cdots & A^{n-1}B
    \end{bmatrix}$ is full-row-rank.
    We use 
    \begin{equation*}\begin{split}
        \phi_R\coloneqq \sigma_n\b{R} > 0    
    \end{split}\end{equation*}
    to denote its minimum nonzero singular value. \qed
\end{assumption}

\subsection{The Naive Least-Squares Estimator}\label{sc:ls}
To better motivate our algorithms, we first analyze the naive least-squares estimator and illustrate the reason why it is biased in our setting. 
For simplicity, we introduce the following auxiliary notations:
\begin{equation*}\begin{split}
    E = \begin{bmatrix}
        A & B
    \end{bmatrix}, \quad z_t = \begin{bmatrix}
        x_t\\
        u_t
    \end{bmatrix}, \quad \htz_t = \begin{bmatrix}
        \htx_t\\
        u_t
    \end{bmatrix}.
\end{split}\end{equation*}
The system dynamics can then be rewritten as $x_{t+1} = Ez_t + w_t$.
Ideally, one would aim to solve for $\arg\min_E \sum_{t=0}^{T-1}\norm{x_{t+1}-Ez_t}^2$. However, because we only have access to noisy observations $\{\htx_{t+1}\}_{t=0}^{T-1}$ and $\{\htz_t\}_{t=0}^{T-1}$, we could only resort to the alternative least-squares problem $\htE_{\ls} = \mathop{\arg\min}_{E}{\sum_{t=0}^{T-1} \norm{\htx_{t+1} - E\htz_t}^2}$
The close-form solution to the alternative problem is 
\begin{equation}\begin{split}\label{eq:estimator_naive}
    \htE_{\ls} = \b{\sum_{t=0}^{T-1}\htx_{t+1}\htz_t\t}\b{\sum_{t=0}^{T-1}\htz_t\htz_t\t}^{-1}.
\end{split}\end{equation}
Notice that $\htx_{t+1} = E\left(\htz_t - \begin{bmatrix}
    \eta_t\\
    0
\end{bmatrix}\right) + w_t +\eta_{t+1}= E\htz_t - A\eta_t + w_t + \eta_{t+1}$. Substituting this into \Cref{eq:estimator_naive} gives
\begin{equation*}\begin{split}
    \htE_{\ls} = {}& E + \b{\sum_{t=0}^{T-1}\b{- A\eta_t + w_t + \eta_{t+1}}\htz_t\t}\b{\sum_{t=0}^{T-1}\htz_t\htz_t\t}^{-1}.
\end{split}\end{equation*}
We note that the noise $- A\eta_t + w_t + \eta_{t+1}$ is correlated with the measurement $\htz_t\t$, since both terms have $\eta_{t}$ as a component. To see how this correlation influences the performance of $\htE_{\ls}$, we further decompose $\htE_{\ls}$ as follows
\begin{equation}\begin{split}\label{eq:estimator_naive2}
    \htE_{\ls} = {}& E + \underbrace{\b{\sum_{t=0}^{T-1}- A\eta_tz_t\t
    + w_t\htz_t\t + \eta_{t+1}\htz_t\t }\b{\sum_{t=0}^{T-1}\htz_t\htz_t\t}^{-1}}_{\Delta_1}\\
    {}& - \underbrace{\b{\sum_{t=0}^{T-1}A\eta_t \begin{bmatrix}
        \eta_t\\
        0
    \end{bmatrix}\t}\b{\sum_{t=0}^{T-1}\htz_t\htz_t\t}^{-1}}_{\Delta_2}.
\end{split}\end{equation}
For term $\Delta_1$, $\eta_t$ is independent of $z_{1:t}$ and $w_t,\eta_{t+1}$ are independent of $\htz_{1:t}$. Leveraging this independency, we can prove by the well-established self-normalizing bound \cite[Theorem 1]{proof_3} that $\norm{\Delta_1} \leq \tilde{\calO}(1/\sqrt{T})$, which decays as the number of samples $T$ grows. In contrast, for term $\Delta_2$, both $\sum_{t=0}^{T-1} \eta_t\eta_t\t$ and $\sum_{t=0}^{T-1}\htz_t\htz_t\t$ are as large as $\calO(T)I$. This implies that $\Delta_2$ doesn't vanish regardless of $T$. This persistent offset term leads to the bias in $\htE_{\ls}$.

To eliminate this bias, we introduce two methods in the following sections, i.e. the \textit{instrumental variable (IV) method} and the \textit{bias compensation (BC) method}, both adapted from previous research on single-input-single-output auto-regressive error-in-variable models \cite{sysid_eiv}.

\section{The Instrumental Variable (IV) Estimator}\label{sc:iv}
\subsection{The IV Estimator and Its Design Rationale }

The IV method eliminates the bias in $\htE_{\ls}$, or equivalently, eliminates the correlation between measurement $\htz_t$ and noise $-A\eta_t+w_t+\eta_{t+1}$, by replacing $\htz_t$ with an instrumental variable
$\hti_t$ and modifying the estimator as follows

\begin{equation}\begin{split}\label{eq:iv_estimator}
    \Aboxed{{}&\htE_{\iv} = \b{\sum_{t=0}^{T-1} \htx_{t+1}\hti_t\t}\b{\sum_{t=0}^{T-1}\htz_t\hti_t\t}^{-1}}\\
    = {}&  E + \b{\sum_{t=0}^{T-1}(- A\eta_t + w_t + \eta_{t+1})\hti_t\t}\b{\sum_{t=0}^{T-1}\htz_t\hti_t\t}^{-1}.
\end{split}\end{equation}
On the high level, we require instrumental variables $\hti_{1:t}$ to be uncorrelated with $-A\eta_t+w_t+\eta_{t+1}$ so that the norm of covariance $\sum_{t=0}^{T-1}(- A\eta_t + w_t + \eta_{t+1})\hti_t\t$ is small\footnote{Asymptotically, for this norm to be small, it only requires $\hti_t$ to be uncorrelated with $-A\eta_t+w_t+\eta_{t+1}$. However, here we aim for a non-asymptotic upper bound and adopt a slightly strong requirement for technical reasons.}. Moreover, we require $\hti_t$ to be well correlated with $\htz_t$ so that the norm of $\sum_{t=0}^{T-1}\htz_t\hti_t\t$ is large. Combining the two requirements would lead to a small error term $\htE_{\iv} - E$. 

It turns out that choosing
\begin{equation}\begin{split}\label{eq:estimator_iv2}
    \hti_t = \begin{bmatrix}
        \htx_{t-1}\\
        u_t
    \end{bmatrix}\in\bbR^{n+m}
\end{split}\end{equation}
suffices.
To see this, we first note that $\hti_{1:t}$ is uncorrelated with $-A\eta_t+w_t+\eta_{t+1}$. We can then prove $\norm{\sum_{t=0}^{T-1}(- A\eta_t + w_t + \eta_{t+1})\hti_t\t}\leq \tilde{\calO}\b{\sqrt{T}}$. 
Secondly, we show $\hti_t$ strongly correlates with measurement $\htz_t$ by analyzing their covariance $\bbE\htz_t\hti_t\t = \bbE\begin{bmatrix}
    \htx_t\htx_{t-1}\t & \htx_tu_t\t\\
    u_t\htx_{t-1}\t & u_tu_t\t
\end{bmatrix}$. 
For the upper left block, we know $\bbE\b{\htx_t\htx_{t-1}\t} = \bbE(Ax_{t-1}+Bu_{t-1}+w_{t-1}+\eta_t)\b{x_{t-1}+\eta_{t-1}}\t = A\bbE\b{x_{t-1}x_{t-1}\t}$ because $u_{t-1}, w_{t-1}, \eta_{t-1}$ and $\eta_t$ are independently sampled from zero-mean distributions. Similarly, we conclude that the off diagonal blocks are $0$. It then naturally follows that, $\bbE\htz_t\hti_t\t = \begin{bmatrix}
    A\bbE\b{x_{t-1}x_{t-1}} & 0\\
    0 & \Sigma_u
\end{bmatrix}$, which is invertible as long as $A$ is invertible. Utilizing this key insight, we can prove that $\sum_{t=0}^{T-1} \htz_t\hti_t\t \succeq \tilde{\calO}(T)I$ for any invertible $A$. Combining the above two inequalities, we can conclude $\norm{\htE_{\iv}-E} \leq \tilde{\calO}(1/\sqrt{T})$, which decays to $0$ as $T$ tends to infinity.



        

\begin{remark}\label{rm:iv}
    If the target system is autonomous, i.e. $x_{t+1} = Ax_t+w_t, \htx_t = x_t+\eta_t$, we can adapt the above IV estimator as follows
    \begin{equation*}\begin{split}
        \htA_{\iv} = \b{\sum_{t=1}^{T-1} \htx_{t+1}\htx_{t-1}\t} \b{\sum_{t=1}^{T-1}\htx_{t}\htx_{t-1}\t}^{-1}.
    \end{split}\end{equation*}
    The theoretical results in the next subsection can be extended to this setting with minor modifications.\qed
\end{remark}


\subsection{Theoretical Guarantee for The IV Estimator}
Now we formally present the finite sample error bound on the IV estimator (\Cref{{eq:iv_estimator}}), along with additional assumptions on the target system $\calM=(A,B,\Sigma_w,\Sigma_u,\Sigma_{\eta})$.
\begin{assumption}\label{assmp1:invertible}
    $A$ is invertible. Let $\phi_A>0$ denote its minimal singular value. \qed
\end{assumption}

\begin{theorem}\label{thm:iv}
    Consider any $\delta\in(0,1)$. Consider system $\calM=(A,B,\Sigma_w,\Sigma_u,\Sigma_{\eta})$ that satisfies \Cref{assmp3:stable} with constants $\psi_A,\rho_A$, \Cref{assmp2:controllable} with constant $\phi_R$ and \Cref{assmp1:invertible} with $\phi_A$. 
    There exist absolute constants $c_1,c_2$ such that
    with probability at least $1-11\delta$
    \begin{equation}\begin{split}
        \norm{\htE_{\iv}-E} \leq {}& c_1\frac{\kappa_1}{\min\{\phi_A,1\}}\sqrt{\frac{m+n}{T}},
    \end{split}\end{equation}
    if $T$ satisfies
    \begin{equation}\begin{split}\label{eq:Tcond_iv}
        T \geq {}& c_2\frac{\kappa_2}{\min\{\phi_A^2, 1\}}n(m+n)^2.
    \end{split}\end{equation}
    Here constants $\kappa_1, \kappa_2$ take the following values
    \begin{equation}\begin{split}\label{eq:kappa}
        \kappa_1 = {}& \psi_A\max\left\{\sqrt{\frac{\psi}{\phi_u}},\frac{\psi}{\phi_u}\right\}\sqrt{\frac{\min\left\{\phi_R^2, 1\right\}\phi_u+1}{\min\left\{\phi_R^6, 1\right\}\phi_u}}\\
        {}& \cdot\sqrt{\log\b{\frac{5\psi\psi_A^2}{(1-\rho_A^2)\delta}n\log\frac{4}{\delta}}},\\
        \kappa_2 = {}& \frac{\psi^2\psi_A^4}{\min\left\{\phi_R^4, 1\right\}\phi_u^2(1-\rho_A^2)}\log\b{\frac{9n}{\delta}}\\
        {}& \cdot\log\b{9\frac{\psi\psi_A^4}{(1-\rho_A^2)\delta}(m+n)\log\frac{9n}{\delta}}.
    \end{split}\end{equation}
    In the above equations $\psi \coloneqq \max\{\psi_B^2\psi_u+\psi_w, \psi_{\eta}\}$, $\psi_B \coloneqq \max\{\norm{B},1\}$, $\psi_u \coloneqq \max\{\norm{\Sigma_u},1\}$, $\psi_w \coloneqq \max\{\norm{\Sigma_w},1\}$, $\psi_\eta \coloneqq \max\{\norm{\Sigma_\eta},1\}$, $\phi_u \coloneqq \sigma_m\b{\Sigma_u}$.\qed

    The error $\norm{\htE_{\iv}-E}$ decays to $0$ as $T$ tends to infinity. Moreover, it scales on the order of $\tilde{\calO}(\sqrt{(m+n)/T})$, matching the error bound to learn a fully observable system without observation noise \cite{sysid_ls1}. Thus, our estimator effectively mitigates the influence of the observation noises while incurring only constant and logarithmic overhead. 

    The estimator imposes the additional assumption that the system matrix $A$ is invertible \Cref{assmp1:invertible}. As illustrated in the previous subsection, this assumption is made to ensure the correlation between the chosen instrumental variable $\hti_t$ and the measurement $\htz_t$. We leave for future work to determine whether this assumption is inevitable for all choices of instrumental variables. 

\end{theorem}

\subsection[]{Proofs for \Cref{thm:iv}}
\begin{proof}
    We begin the proof by decomposing the error $\htE_{\iv}-E$ into two terms: $\Delta_4$ which captures the covariance between $\hti_t$ and $\eta_{t+1}+w_t-A\eta_t$, and $\Delta_5$ which captures the inverse of the covariance between $\hti_t$ and $\htz_t$ (\textit{Step 1}). We then upper bound  $\Delta_4$ (\textit{Step 2}) and $\Delta_5$ (\textit{Step 3}). Combining the two bounds gives the desired result (\textit{Step 4}).

    \textit{Step 1: Error Decomposition.} Recall that $\htz_t=\begin{bmatrix}
        \htx_t\\
        u_t
    \end{bmatrix}$ and $\hti_t = \begin{bmatrix}
        \htx_{t-1}\\
        u_t
    \end{bmatrix}$. We define $i_t \coloneqq \begin{bmatrix}
        x_{t-1}\\
        u_t
    \end{bmatrix}$.
    By definition, 
    \begin{equation*}\begin{split}
        {}& \htE_{\iv}
        = \b{\sum_{t=1}^{T-1} \htx_{t+1}\hti_t\t} \b{\sum_{t=1}^{T-1} \htz_t\hti_t\t}^{-1}\\
        ={}& E + \underbrace{\b{\sum_{t=1}^{T-1} \b{\eta_{t+1}+w_t-A\eta_t}\hti_t\t}\b{\sum_{t=1}^{T-1} \htz_t\hti_t\t}^{-1}}_{\Delta_3}.
    \end{split}\end{equation*}
    To upper bound $\norm{\htE_{\iv}-E}=\norm{\Delta_3}$, we decompose it as
    {\small
    \begin{equation}\begin{split}\label{eq:iv_11}
        \norm{\Delta_3} \leq {}& \Big\|\underbrace{\sum_{t=1}^{T-1} \b{\eta_{t+1}+w_t-A\eta_t}\hti_t\t\Sigma_i^{-1}}_{\Delta_4}\Big\| \Big\|\underbrace{\Sigma_i\big(\sum_{t=1}^{T-1} \htz_t\hti_t\t\big)^{-1}}_{\Delta_5}\Big\|,
    \end{split}\end{equation}
    }where we have defined $\Sigma_i\coloneqq \sum_{t=1}^{T-1} i_ti_t\t$. The invertibility of $\Sigma_i$ is guaranteed by the following lemma.
    \begin{lemma}\label{lem:iv2}
        Consider the setting of \Cref{thm:iv}. Suppose $T$ satisfies \Cref{eq:Tcond_iv}.
        Then for any $\delta\in(0,1)$, the following holds with probability at least $1-\delta$,
        \begin{equation*}\begin{split}
            \pushQED{\qed} 
            \min\left\{\frac{\phi_R^2}{16}, \frac{3}{8}\right\}\phi_uTI \preceq \Sigma_i \precsim \frac{\psi\psi_A^2}{1-\rho_A^2}\log\b{\frac{4}{\delta}}nTI.\qedhere
            \popQED
        \end{split}\end{equation*}
    \end{lemma}
    \textit{Step 2: Upper bound of $\mathit{\norm{\Delta_4}}$.}
    We decompose $\Delta_4$ as
    \begin{equation}\begin{split}\label{eq:iv_5}
        \norm{\Delta_4} \leq {}& \norm{\sum_{t=1}^{T-1} A\eta_{t}\hti_t\t\Sigma_i^{-1}} + \norm{\sum_{t=1}^{T-1} \b{\eta_{t+1}+w_t}i_t\t\Sigma_i^{-1}}\\
        + {}& \norm{\sum_{t=1}^{T-1} \b{\eta_{t+1}+w_t}\begin{bmatrix}
            \eta_{t-1}\\
            0
        \end{bmatrix}\t\Sigma_i^{-1}}
    \end{split}\end{equation}
    We will use the following lemma, an adapted version of the well-established self-normalizing bound (Theorem 1 in \cite{proof_3}), to bound the above three terms.
    \begin{lemma}\label{lem:iv3}
        Let $\{y_t\}_{t=1}^T$ be an arbitrary $\bbR^a$-valued stochastic process. Consider any positive integer $b$, and let $\{\zeta_t\in\bbR^b\}_{t=1}^T$ be a sequence of i.i.d. Gaussian vectors from $\calN(0, \Sigma_\zeta)$ such that $\zeta_t$ is independent of $y_{1:t}$. 
        The following holds with probability at least $1-\delta$ for any $\delta\in(0,1)$,
        \begin{equation}\begin{split}\label{eq:iv3_1}
            {}& \norm{\b{\overline{\Sigma}_y}^{-\frac{1}{2}} \sum_{t=1}^T y_t\zeta_{t}\t}\\
            \lesssim {}& \sqrt{\max\{a,b\}\psi_\zeta\log\b{\frac{5(\sigma_1/T+1)}{\delta}}}.
        \end{split}\end{equation}
        Here $\Sigma_y \coloneqq \sum_{t=1}^Ty_ty_t\t$, $\overline{\Sigma}_y \coloneqq \Sigma_y+TI$, $\sigma_1 \coloneqq \sigma_1\b{\Sigma_y}$ and $\psi_{\zeta} \coloneqq \max\{\sigma_1\b{\Sigma_\zeta},1\}$.

        If $\Sigma_y$ is invertible, the following holds with probability at least $1-\delta$ for any $\delta\in(0,1)$,
        \begin{equation}\begin{split}\label{eq:iv3_2}
            {}& \norm{\b{\Sigma_y}^{-1} \sum_{t=1}^T y_t\zeta_{t}\t}\\
            \lesssim {}& \frac{\sqrt{\sigma_n+T}}{\sigma_n}\sqrt{\max\{a,b\}\psi_\zeta\log\b{\frac{5(\sigma_1/T+1)}{\delta}}}.
        \end{split}\end{equation}
        \qed
    \end{lemma}
    To bound the second term in \Cref{eq:iv_5}, we notice that $\eta_{t+1}+w_t\overset{\text{i.i.d.}}{\sim}\calN(0,\Sigma_\eta+\Sigma_w)$ and is independent of $i_{1:t}$, i.e. independent of $u_{1:t}$ and $x_{0:t-1}$. Moreover, by \Cref{lem:iv2}, with probability at least $1-\delta$
    \begin{equation}\begin{split}\label{eq:iv_99}
        \sigma_1\b{\Sigma_i} \precsim {}&\frac{\psi\psi_A^2}{1-\rho_A^2}\log\b{\frac{4}{\delta}}nT\\ \sigma_n(\Sigma_i)\geq {}& \min\left\{\frac{\phi_R^2}{16}, \frac{3}{8}\right\}\phi_uT.
    \end{split}\end{equation}
    We can then apply \Cref{eq:iv3_2} on $\{i_{t}\}_{t=1}^T$, $\{\eta_{t+1}+w_t\}_{t=1}^T$ and get the following with probability at least $1-2\delta$
    \begin{equation}\begin{split}\label{eq:iv_4}
        {}& \norm{\b{\Sigma_{i}}^{-1} \sum_{t=1}^{T-1}i_t\b{\eta_{t+1}+w_t}\t}\\
        \lesssim {}& \frac{\sqrt{\min\left\{\phi_R^2, 6\right\}\phi_u+16}}{\min\left\{\phi_R^2, 6\right\}\phi_u}\\
        {}& \sqrt{\frac{m+n}{T}\b{\psi_\eta+\psi_w}\log\b{\frac{5\psi\psi_A^2}{(1-\rho_A^2)\delta}n\log\frac{4}{\delta}}}.
    \end{split}\end{equation}

    Similarly, for the third term in \Cref{eq:iv_5}, we note that $\eta_{t+1}+w_t\overset{\text{i.i.d.}}{\sim}\calN(0,\Sigma_\eta+\Sigma_w)$ and is independent of $\eta_{0:t-1}$. Moreover, we define $\Sigma_{\tilde{\eta}}\coloneqq\sum_{t=0}^{T-2}\begin{bmatrix}
        \eta_t\\
        0
    \end{bmatrix}\begin{bmatrix}
        \eta_t\\
        0
    \end{bmatrix}\t$ and get the following by a standard covariance concentration result (e.g. Equation (45) of \cite{proof_2}) with probability at least $1-\delta$ for $T$ satisfying \Cref{eq:Tcond_iv}
    \begin{equation*}\begin{split}
        \sigma_1\b{\Sigma_{\tilde{\eta}}} \leq \frac{5}{4}\psi_\eta T,\quad 
        \sigma_{m+n}\b{\Sigma_{\tilde{\eta}}} \geq \frac{3}{4}\phi_\eta T.
    \end{split}\end{equation*}
    We can then apply \Cref{eq:iv3_1} on $\{\begin{bmatrix}
        \eta_{t-1}\\
        0
    \end{bmatrix}\}_{t=1}^T$, $\{\eta_{t+1}+w_t\}_{t=1}^T$ and get the following with probability at least $1-2\delta$
    \begin{equation*}\begin{split}
        {}& \norm{\sum_{t=1}^{T-1}\begin{bmatrix}
            \eta_{t-1}\\
            0
        \end{bmatrix}\b{\eta_{t+1}+w_t}\t}\\
        \lesssim {}& \sqrt{\norm{\Sigma_{\tilde{\eta}}+TI}}\sqrt{(m+n)\b{\psi_\eta+\psi_w}\log\b{\frac{6\psi_{\eta}}{\delta}}}\\
        \lesssim {}& \sqrt{(m+n)\psi_{\eta}\b{\psi_\eta+\psi_w}T\log\b{\frac{6\psi_{\eta}}{\delta}}}.
    \end{split}\end{equation*}
    Therefore, combining \Cref{eq:iv_99}, we have the following with probability at least $1-3\delta$
    \begin{equation}\begin{split}\label{eq:iv_7}
        {}& \norm{\Sigma_i^{-1}\sum_{t=1}^{T-1}\begin{bmatrix}
            \eta_{t-1}\\
            0
        \end{bmatrix}\b{\eta_{t+1}+w_t}\t}\\
        \lesssim {}& \frac{1}{\min\{\phi_R^2,6\}\phi_u}\sqrt{\frac{m+n}{T}\psi_{\eta}\b{\psi_\eta+\psi_w}\log\b{\frac{6\psi_{\eta}}{\delta}}}.
    \end{split}\end{equation}
    Combining \Cref{eq:iv_4,eq:iv_7} with a union bound gives the following with probability at least $1-5\delta$
    \begin{equation*}\begin{split}
        \norm{\Sigma_i^{-1}\sum_{t=1}^{T-1}\hti_t\b{\eta_{t+1}+w_t}\t}
        \lesssim {}& \kappa_1\sqrt{\frac{m+n}{T}}.
    \end{split}\end{equation*}
    Here $\kappa_1 = \psi_A\sqrt{\frac{\min\left\{\phi_R^2, 1\right\}\phi_u+1}{\min\left\{\phi_R^6, 1\right\}\phi_u}\log\b{\frac{5\psi\psi_A^2}{(1-\rho_A^2)\delta}n\log\frac{4}{\delta}}}$\\
    $\cdot\max\left\{\sqrt{\frac{\psi}{\phi_u}},\frac{\psi}{\phi_u}\right\}$. Similarly, with probability at least $1-5\delta$,
    \begin{equation*}\begin{split}
        \norm{\Sigma_i^{-1} \sum_{t=1}^{T-1} \hti_t\eta_{t}\t A\t}\lesssim {}& \kappa_1\sqrt{\frac{m+n}{T}}.
    \end{split}\end{equation*}
    
    Combining the above two inequalities with a union bound, we get the following with probability at least $1-10\delta$
    \begin{equation}\begin{split}\label{eq:iv_10}
        \norm{\Delta_4} 
        \leq {}& \norm{\Sigma_i^{-1}\sum_{t=1}^{T-1}\hti_t\b{\eta_{t+1}+w_t}\t} + \norm{\Sigma_i^{-1}\sum_{t=1}^{T-1} \hti_t\eta_{t}\t A\t}\\
        \lesssim {}& \kappa_1\sqrt{\frac{m+n}{T}}.
    \end{split}\end{equation}

    \textit{Step 3: Upper bound of $\norm{\Delta_5}$}. We achieve this upper bound by lower bounding its inverse by the following lemma
    \begin{lemma}\label{lem:iv1}
        Consider the setting of \Cref{thm:iv}. Suppose $T$ satisfies \Cref{eq:Tcond_iv}.
        Then with probability at least $1-\delta$ for any $\delta\in(0,1)$,
        we know that
        \begin{equation*}\begin{split}
            \pushQED{\qed} 
            \sigma_{m+n}\b{\b{\sum_{t=0}^{T-1}\htz_t\hti_t\t}\Sigma_i^{-1}} \geq \frac{\min\{\phi_A,1\}}{2}.\qedhere
            \popQED
        \end{split}\end{equation*}
    \end{lemma}
    As a direct consequence, with probability at least $1-\delta$,
    \begin{equation}\begin{split}\label{eq:iv_9}
        \norm{\Delta_5} \leq \frac{2}{\min\{\phi_A,1\}}.
    \end{split}\end{equation}
    
    \textit{Step 4: Error Upperbound $\mathit{\norm{\htE_{\iv}-E}}$.}
    Finally, substituting \Cref{eq:iv_9,eq:iv_10} back into \Cref{eq:iv_11}, we get the following with probability at least $1-11\delta$
    \begin{equation*}\begin{split}
        \norm{\htE_{\iv}-E} = {}& \norm{\Delta_3} \leq \norm{\Delta_4}\norm{\Delta_5}\\
        \lesssim {}& \frac{\kappa_1}{\min\{\phi_A,1\}}\sqrt{\frac{m+n}{T}}.
    \end{split}\end{equation*}
    This finishes the proof.
\end{proof}

\section{The Bias Compensation (BC) Estimator}\label{sc:bc}
\subsection{The BC Estimator and Its Design Rationale}
In the bias compensation method, we assume access to an estimation of the observation noise covariance $\wh{\Sigma}_{\eta}$ that satisfies the following inequality for some constant $\epsilon_{\eta}$
\begin{equation}\label{eq:cov_def}
    \norm{\wh{\Sigma}_{\eta}-\Sigma_{\eta}} \leq \epsilon_{\eta}.
\end{equation}
With $\wh{\Sigma}_{\eta}$, we can directly estimate the bias in the naive least-squares estimator $\wh{E}_{\ls}$ and compensate for it by designing a multiplicative correction term.

To derive the correction term, we take a deeper look into $\htE_{\ls}$ (\Cref{eq:estimator_naive2})
\begin{equation*}\begin{split}    
    \htE_{\ls} = {}& E + \Delta_1 - E\b{\sum_{t=0}^{T-1} \begin{bmatrix}
        \eta_t\\
        0
    \end{bmatrix}\begin{bmatrix}
        \eta_t\\
        0
    \end{bmatrix}\t}\b{\sum_{t=0}^{T-1}\htz_t\htz_t\t}^{-1}\\
    = {}& E + \Delta_1 - E\begin{bmatrix}
        \Sigma_{\eta} & 0\\
        0 & 0
    \end{bmatrix}\b{\frac{1}{T}\sum_{t=0}^{T-1}\htz_t\htz_t\t}^{-1}\\
    {}& + E\begin{bmatrix}
        \Sigma_{\eta}-\frac{1}{T}\sum_{t=0}^{T-1}\eta_t\eta_t\t & 0\\
        0 & 0
    \end{bmatrix}\b{\frac{1}{T}\sum_{t=0}^{T-1}\htz_t\htz_t\t}^{-1}\\
\end{split}\end{equation*}
As we discussed in \Cref{sc:ls}, we can show that $\norm{\Delta_1}\leq \tilde{\calO}(1/\sqrt{T})$. Moreover, by a standard covariance concentration lemma, e.g. \cite[Lemma A.5]{proof_1}, we can show that $\|\Sigma_\eta-\frac{1}{T}\sum_{t-0}^{T-1}\eta_t\eta_t\| \leq \tilde{\calO}(1/\sqrt{T})$. Therefore, we can rewrite $\htE_{\ls}$ as follows
{\small\begin{equation*}\begin{split}    
    \htE_{\ls} = {}& E\b{I - \begin{bmatrix}
        \Sigma_\eta & 0\\
        0 & 0
    \end{bmatrix}\b{\frac{1}{T}\sum_{t=0}^{T-1}\htz_t\htz_t\t}^{-1}} + \tilde{\calO}\b{\frac{1}{\sqrt{T}}}I.
\end{split}\end{equation*}}

Therefore, with $\wh{\Sigma}_{\eta}$ (\Cref{eq:cov_def}), we can compensate for the bias as follows
\begin{empheq}[box=\fbox]{align}\label{eq:bc_estimator}
    \htE_{\bc} = \htE_{\ls} \b{I - \begin{bmatrix}
        \wh{\Sigma}_\eta & 0\\
        0 & 0
    \end{bmatrix}\b{\frac{1}{T}\sum_{t=0}^{T-1} \htz_t\htz_{t}\t}^{-1}}^{-1}.
\end{empheq}




        
\begin{remark}\label{rm:bc}
    If the target system is autonomous, i.e. $x_{t+1} = Ax_t+w_t, \htx_t = x_t+\eta_t$, we can easily adapt the above estimator as follows
    \begin{equation*}\begin{split}
        \htA_{\bc} = {}& \htA_{\ls} \b{I -  \wh{\Sigma}_\eta\b{\frac{1}{T}\sum_{t=0}^{T-1} \htx_t\htx_{t}\t}^{-1}}^{-1},\\
        \htA_{\ls} = {}& \sum_{t=0}^{T-1} \htx_{t+1}\htx_{t}\t \b{\sum_{t=0}^{T-1}\htx_{t}\htx_{t}\t}^{-1}.
    \end{split}\end{equation*}
    Our proof techniques can be extended to this setting with minor modifications.\qed
\end{remark}

In the BC estimator, the invertibility of matrix $I - \begin{bmatrix}
    \wh{\Sigma}_\eta & 0\\
    0 & 0
\end{bmatrix}\b{\frac{1}{T}\sum_{t=0}^{T-1} \htz_t\htz_{t}\t}^{-1}$ can be ensured if the minimal singular value $\sigma_{m+n}\b{\frac{1}{T}\sum_{t=0}^{T-1}\htz_t\htz_t\t}$ is large. Intuitively, this can be easily satisfied by injecting large enough inputs $\{u_t\}_{t=0}^{T-1}$ into the system. The exact requirements are summarized in \Cref{assmp:input}.
\subsection{Theoretical Guarantee}
Now we present the error upper bound for the BC estimator (\Cref{eq:bc_estimator}), along with necessary assumptions.
\begin{assumption}\label{assmp:input}
    Consider system $\calM=(A,B,\Sigma_w,\Sigma_u,\Sigma_{\eta})$ that satisfies \Cref{assmp2:controllable} with constant $\phi_R$. Suppose $\sigma_{m}\b{\Sigma_u} \coloneqq \phi_u$ satisfies
    \begin{equation}\begin{split}
        \phi_u \geq \frac{32\b{\psi_\eta+\epsilon_{\eta}}}{\min\{\phi_R^2,6\}}.
    \end{split}\end{equation}
    Here $\psi_{\eta} \coloneqq \max\{\sigma_1\b{\Sigma_{\eta}},1\}$, and $\epsilon_{\eta}$ is the error of the observation noise covariance estimation $\wh{\Sigma}_{\eta}$ defined in \Cref{eq:cov_def}.
\end{assumption}


\begin{theorem}\label{thm:bc}
    Consider system $\calM=(A,B,\Sigma_w,\Sigma_u,\Sigma_{\eta})$ satisfying \Cref{assmp3:stable} with constants $\psi_A,\rho_A$, \Cref{assmp2:controllable} with constant $\phi_R$ and \Cref{assmp:input}. Suppose we have access to $\wh{\Sigma}_{\eta}$ that satisfies $\norm{\wh{\Sigma}_{\eta}-\Sigma_\eta}\leq \epsilon_{\eta}$.
    There exist absolute constants $c_1,c_2$ such that
    with probability at least $1-9\delta$
    \begin{equation}\begin{split}
        {}& \norm{\htE_{\bc} - E} 
        \leq c_1\frac{\epsilon_{\eta}}{\min\{\phi_R^2,1\}\phi_u} + c_1\kappa_1\sqrt{\frac{m+n}{T}}.
    \end{split}\end{equation}
    if $T$ satisfies
    \begin{equation}\begin{split}\label{eq:Tcond_bc}
        T \geq {}& c_2\kappa_2n(m+n)^2.
    \end{split}\end{equation}
    Here the constants $\kappa_1, \kappa_2$ are the constant in \Cref{eq:kappa} (\Cref{thm:iv}).\qed
\end{theorem}

The error bound scales linearly in the approximation error $\epsilon_\eta$ (\Cref{eq:cov_def}) of observation noise covariance error, and its influence can be mitigated by injecting larger inputs into the system. In the case where $\epsilon_{\eta}\leq \tilde{\calO}(\sqrt{(m+n)/T})$, that is, when the observation noise covariance estimation is accurate, the BC estimator achieves an error bound $\norm{\htE_{\bc}-E}=\tilde{\calO}(\sqrt{(m+n)/T})$. This matches the performance of the IV estimator \Cref{thm:iv} without requiring the system matrix $A$ to be invertible.

Due to the space limit, we leave the proof of \Cref{thm:bc} to Appendix II.

\section{Simulations}\label{sc:sim}
We perform simulations in both a non-autonomous system with inputs and an autonomous system. 
Specifically, the non-autonomous system has the following parameters:
\begin{equation*}\begin{split}
    {}& A = 0.8*\begin{bmatrix}
        0 & 1\\
        I_{19} & 0
    \end{bmatrix}, B = \begin{bmatrix}
        I_{10}\\
        0\\
    \end{bmatrix}, \Sigma_w=\Sigma_{\eta}=I_{20}, \Sigma_u=I_5.
\end{split}\end{equation*}
With a single trajectory from the above system with length $T$, we computed the naive least-squares estimator (\Cref{eq:estimator_naive}), the Ho-Kalman estimator \cite{sysid_1}, the IV estimator (\Cref{eq:iv_estimator}) and the BC estimator (\Cref{eq:bc_estimator}).
As is shown in \Cref{fig:result}, errors of the proposed IV and BC estimators converge to almost zero as $T$ while the naive least-squares estimator remains biased. Moreover, the proposed algorithms converge much faster than the Ho-Kalman algorithm in this setting.

The autonomous system shares the same $A, \Sigma_w, \Sigma_\eta$ and $\Sigma_u$. With a single trajectory of length $T$ from the system, we computed the naive least-squares estimator (\Cref{rm:bc}), the IV estimator (\Cref{rm:iv}) and the BC estimator (\Cref{rm:bc}). The errors of the proposed IV and BC estimators converge to almost $0$ while the naive least-squares estimator remains biased.
\begin{figure}[h]\begin{center}
    \begin{minipage}[h]{0.49\linewidth}
       \centering
       \includegraphics[width=\linewidth]{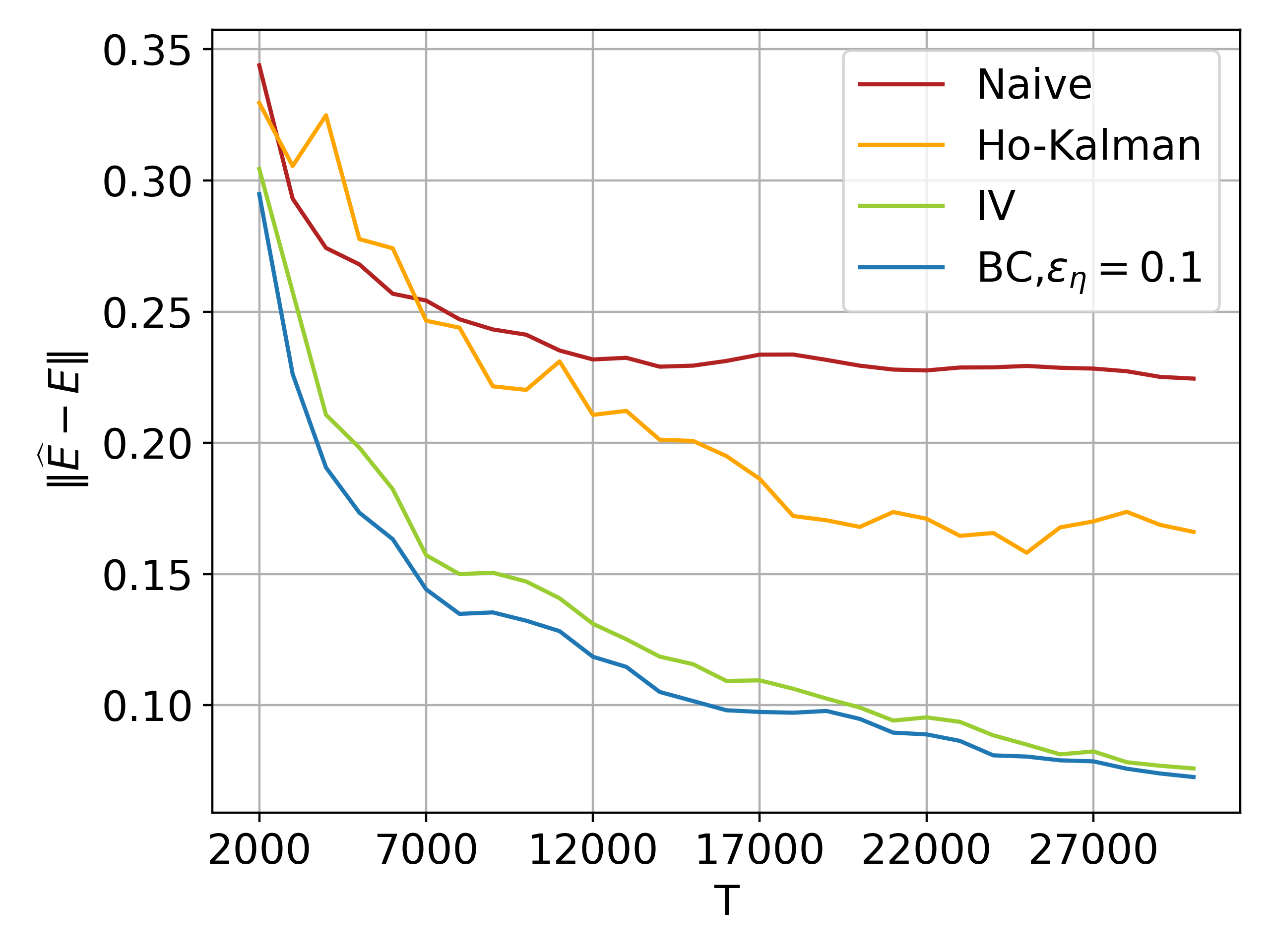}
    \end{minipage}
    \begin{minipage}[h]{0.49\linewidth}
       \centering
       \includegraphics[width=\linewidth]{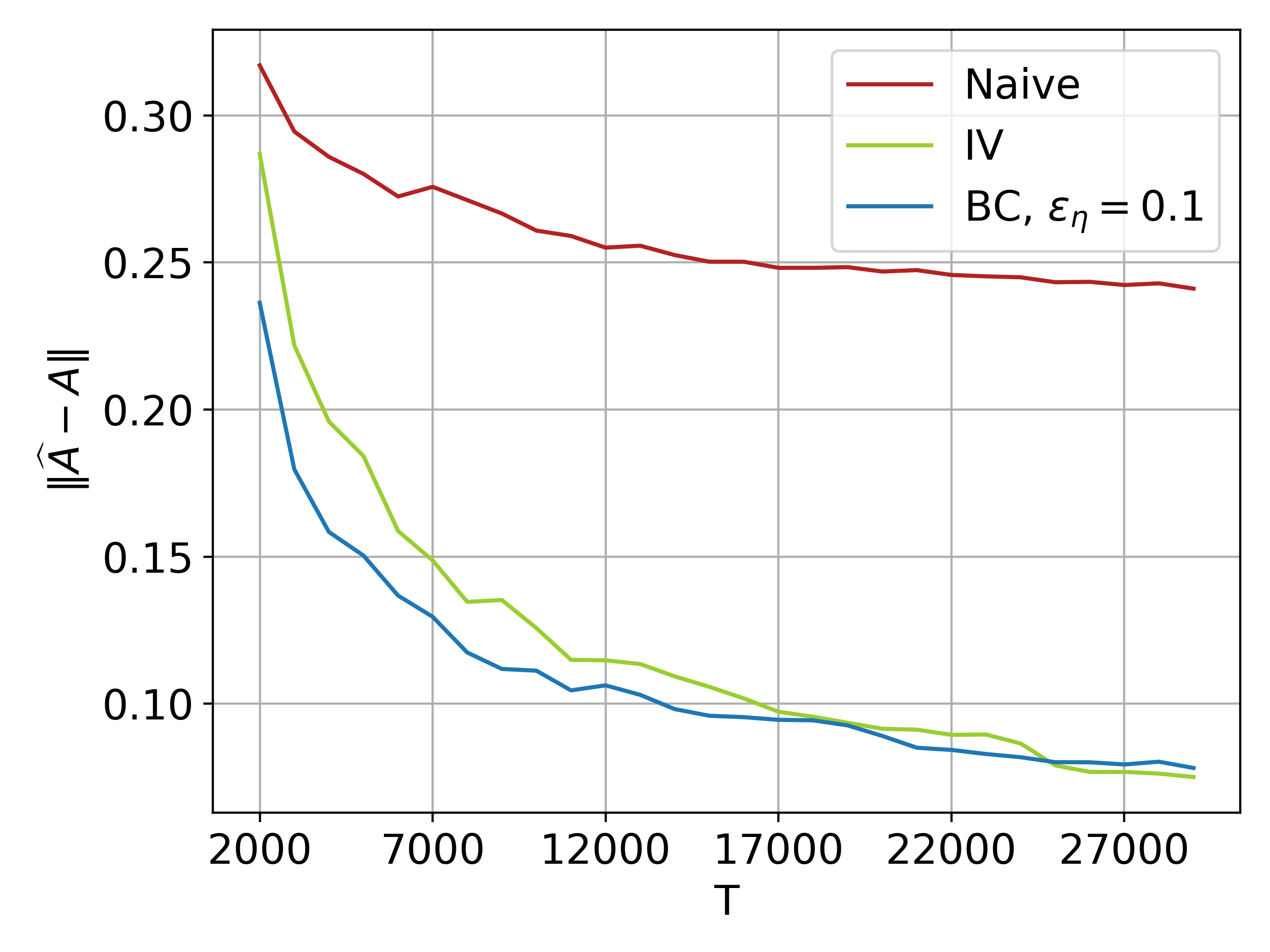}
    \end{minipage}
    \caption{\textit{Left:} Estimation error for the non-autonomous system; \textit{Right:} Estimation error for the autonomous system.}
    \label{fig:result}
\end{center}\end{figure}

\vspace{-10pt}\section{Conclusion}
In this paper, we focus on learning linear dynamical systems with observation noises. We propose two estimators and provide non-asymptotic guarantee of their performance. The IV estimator requires the dynamic matrix $A$ to be invertible and achieves an error upper bound of order $\tilde{\calO}(\sqrt{(m+n)/T})$. This order matches the best-known error upper bound for learning a fully observable system without observation noises. The BC estimator requires an estimation of the observation noise covariance and achieves an error upper bound of order $\tilde{\calO}(\sqrt{(m+n)/T}+\epsilon_{\eta})$, where $\epsilon_{\eta}$ is the error of the observation noise covariance estimation. 



\bibliography{ref}

\newpage\onecolumn
\appendices
\section{Proofs for \Cref{sc:iv}}
\begin{proof}[Proof of \Cref{lem:iv2}]
    Throughout this proof, we consider $\delta'\in(0,\frac{1}{3})$, which ensures $1\leq \log(k/\delta') \leq k\log(1/\delta')$ for all positive integer $k$.
    By definition, we have
    \begin{equation}\begin{split}\label{eq:covi_1}
        \sum_{t=1}^{T-1} i_ti_t\t = {}& \sum_{t=1}^{T-1} \begin{bmatrix}
            x_{t-1}\\
            u_t
        \end{bmatrix}\begin{bmatrix}
            x_{t-1}\\
            u_t
        \end{bmatrix}\t= \sum_{t=1}^{T-1} \begin{bmatrix}
            x_{t-1}x_{t-1}\t & 0\\
            0 & u_tu_t\t
        \end{bmatrix} + \begin{bmatrix}
            0 & x_{t-1}u_t\t\\
            u_tx_{t-1}\t & 0
        \end{bmatrix}.
    \end{split}\end{equation}

    We first show the signal terms, i.e. $\sum_{t=1}^{T-1} x_{t-1}x_{t-1}\t$ and $\sum_{t=1}^{T-1} u_tu_t\t$ are large. By a standard covariance concentration result (Equation (45) of \cite{proof_2}), we have the following inequalities for $T \gtrsim n + \log\b{\frac{1}{\delta'}}$, each with probability at least $1-\delta'$
\begin{equation}\begin{split}\label{eq:covi_5}
    \frac{3}{4}(T-1)\Sigma_u \preceq \sum_{t=1}^{T-1} u_tu_t\t \preceq \frac{5}{4}(T-1)\Sigma_u.
\end{split}\end{equation}
Moreover, by a state covariance lower bound result (Lemma A.3 in \cite{proof_1}) and upper bound result (\Cref{lem:cov_upper}), with probability at least $1-2\delta'$
\begin{equation}\begin{split}\label{eq:covi_6}
    \frac{\phi_u\phi_R^2}{8}(T-2)I \preceq \sum_{t=1}^{T-1} x_{t-1}x_{t-1}\t \precsim \frac{\psi\psi_A^2n}{1-\rho_A^2}TI\log\frac{1}{\delta'}.
\end{split}\end{equation}
for any $T$ satisfies
\begin{equation}\begin{split}
    T \gtrsim\frac{\psi^2\psi_A^4}{\phi_u^2\phi_R^4}n^3\log\b{\frac{n}{\delta'}}\log\b{\frac{\psi\psi_A^4}{1-\rho_A^2}n\log\frac{n}{\delta'}}.
\end{split}\end{equation}

Now we prove that the cross terms are small. We note that $\{x_{t-1}\}_{t=1}^{T-1}$ is a trajectory from the target system $\calM = (A,B,\Sigma_w, \Sigma_u,\Sigma_{\eta})$, and that $\{u_{t}\}_{t=1}^{T-1}$ a sequence of i.i.d gaussian vectors such that $u_t$ is independent of $x_{0:t-1}$. Therefore, we can apply \Cref{cor:cross_general}, an adapted version of the well-established self-normalizing bound (Theorem 1 in \cite{proof_3}), on the two sequences and get the following with probability at least $1-\delta'$
\begin{equation}\begin{split}\label{eq:covi_3}
    \norm{\sum_{t=1}^{T-1} x_{t-1}u_{t}\t } \lesssim {}& \sqrt{\frac{\psi\psi_A^2\psi_u}{1-\rho_A^2}n\max\{m,n\}T\log\b{\frac{1}{\delta'}}}\sqrt{\log\b{\frac{\psi\psi_A^2}{(1-\rho_A^2)\delta'}n\log\frac{1}{\delta'}}}.
\end{split}\end{equation}

Substituting \Cref{eq:covi_3,eq:covi_5,eq:covi_6} back into \Cref{eq:covi_1} gives the following for some absolute constant $c_1$ with probability at least $1-4\delta'$
\begin{equation}\begin{split}\label{eq:resvi}
    \sum_{t=1}^{T-1} i_ti_t\t \succeq {}& \sum_{t=1}^{T-1} \begin{bmatrix}
        x_{t-1}x_{t-1}\t & 0\\
        0 & u_tu_t\t
    \end{bmatrix} - \norm{\sum_{t=1}^{T-1}\begin{bmatrix}
        0 & x_{t-1}u_t\t\\
        u_tx_{t-1}\t & 0
    \end{bmatrix}}I\\
    \succeq {}& \min\left\{\frac{\phi_R^2}{8}, \frac{3}{4}\right\}\phi_u(T-2)I - c_1\sqrt{\frac{\psi^2\psi_A^2}{1-\rho_A^2}(m+n)^2\log\b{\frac{1}{\delta'}}}\sqrt{\log\b{\frac{\psi\psi_A^2}{(1-\rho_A^2)\delta'}n\log\frac{1}{\delta'}}}\sqrt{T}\\
    \succeq {}& \min\left\{\frac{\phi_R^2}{16}, \frac{3}{8}\right\}\phi_uTI.
\end{split}\end{equation}
Here the last line is by \Cref{eq:Tcond_iv} with $\delta=4\delta'$. Moreover,
\begin{equation}\begin{split}
    \sum_{t=1}^{T-1} i_ti_t\t \preceq {}& \norm{\sum_{t=1}^{T-1} \begin{bmatrix}
        x_{t-1}x_{t-1}\t & 0\\
        0 & u_tu_t\t
    \end{bmatrix}}I + \norm{\sum_{t=1}^{T-1}\begin{bmatrix}
        0 & x_{t-1}u_t\t\\
        u_tx_{t-1}\t & 0
    \end{bmatrix}}I\\
    \precsim {}& \max\left\{\frac{\psi\psi_A^2n}{1-\rho_A^2}T\log\frac{1}{\delta'}, \frac{5}{4}\psi_{u}T\right\}I + c_1\sqrt{\frac{\psi^2\psi_A^2}{1-\rho_A^2}(m+n)^2\log\b{\frac{1}{\delta'}}}\sqrt{\log\b{\frac{\psi\psi_A^2}{(1-\rho_A^2)\delta'}n\log\frac{1}{\delta'}}}\sqrt{T}\\
    \precsim {}& \frac{\psi\psi_A^2n}{1-\rho_A^2}TI\log\frac{1}{\delta'}.
\end{split}\end{equation}
Here the last line is by \Cref{eq:Tcond_iv} with $\delta=4\delta'$. 
Finally, let $\delta = 4\delta'$, and we get the desired result.
\end{proof}

\begin{proof}[Proof of \Cref{lem:iv3}]
    From Lemma A.9 in \cite{proof_1}, we know that the following holds for some vector $v\in\mathbb{S}^{b-1}$
    \begin{equation}\label{eq:thm1_1}\begin{split}
        \bbP\b{\norm{\b{\overline{\Sigma}_y}^{-\frac{1}{2}} \sum_{t\in[T]} y_t\zeta_{t}\t}>z} \leq 5^b \bbP\b{\norm{\b{\overline{\Sigma}_y}^{-\frac{1}{2}} \sum_{t\in[T]} y_t\zeta_{t}\t v}>\frac{z}{2}}.
    \end{split}\end{equation}
    Notice that $\zeta_{t}\t v$ are independent Gaussian variables from distribution $\calN(0, v\t \Sigma_{\zeta}v)$ for all $t\in[T]$, which is $c_{1}\sqrt{\psi_\zeta}$-subGaussian for some absolute constant $c_{1}$. Then with filtration $\{\calF_t = \sigma\b{y_1, \cdots, y_{t+1}, \eta_1, \cdots, \eta_t}\}_{t\in[T]}$\footnote{Here we slightly abuse the notation and let $\sigma(\cdot)$ denote the sigma-algebra generated by $\cdot$.}, we know that $y_t$ is $\calF_{t-1}$-measurable while $\eta_t\t v$ is $\calF_t$-measurable and sub-gaussian conditioned on $\calF_{t-1}$. We can then apply Theorem 1 in \cite{proof_3} on sequence $\{y_t\}_{t\in[T]}$ and sequence $\{\zeta_{t}\t v\}_{t\in[T]}$, gives the following inequality for any $\delta'\in(0,1)$
    \begin{equation}\begin{split}
        \bbP\b{\norm{\b{\overline{\Sigma}_y}^{-\frac{1}{2}} \sum_{t\in[T]} y_t\zeta_{t}\t v} >  \sqrt{2c_{1}^2\psi_\zeta  \log\b{\frac{\det(\overline{\Sigma}_y)^{\frac{1}{2}}\det( T  I)^{-\frac{1}{2}}}{\delta'}}}} \leq \delta'.
    \end{split}\end{equation}
    Substituting the above result back into \Cref{eq:thm1_1} gives the following inequality
    \begin{equation}\begin{split}
        \bbP\b{\norm{\b{\overline{\Sigma}_y}^{-\frac{1}{2}} \sum_{t\in[T]} y_t\zeta_{t}\t} >  \sqrt{8c_{1}^2\psi_\zeta  \log\b{\frac{\det(\overline{\Sigma}_y)^{\frac{1}{2}}\det( T  I)^{-\frac{1}{2}}}{\delta'}}}} \leq 5^b\delta',
    \end{split}\end{equation}
    which implies the following inequality holds with probability at least $1-\delta$ for any $\delta=5^b\delta'\in(0,1)$
    \begin{equation}\label{eq:cross_general_2}\begin{split}
        \norm{\b{\overline{\Sigma}_y}^{-\frac{1}{2}} \sum_{t\in[T]} y_t\zeta_{t}\t} \leq {}& \sqrt{8c_{1}^2\psi_\zeta\log\b{\frac{2\det(\overline{\Sigma}_y)^{\frac{1}{2}}\det( T  I)^{-\frac{1}{2}}}{\delta}} + 8c_{1}^2\psi_\zeta b\log5}\\
        \lesssim {}& \sqrt{\psi_\zeta\log\b{\frac{2\det(\overline{\Sigma}_y)^{\frac{1}{2}}\det( T  I)^{-\frac{1}{2}}}{\delta}} + \psi_\zeta b\log5}\\
        \leq {}& \sqrt{\psi_\zeta\log\b{\frac{2(\sigma_1(\sum_{t=1}^{T}y_ty_t\t)/T+1)^{\frac{n}{2}}}{\delta}} + \psi_\zeta b\log5}\\
        \leq {}& \sqrt{n\psi_\zeta\log\b{\frac{2(\sigma_1(\sum_{t=1}^{T}y_ty_t\t)/T+1)}{\delta}} + \psi_\zeta b\log5}\\
        \lesssim {}& \sqrt{\max\{a,b\}\psi_\zeta\log\b{\frac{5(\sigma_1(\sum_{t=1}^{T}y_ty_t\t)/T+1)}{\delta}}}.
    \end{split}\end{equation}
    This finishes the proof of \Cref{eq:iv3_1}

    \textbf{Now we prove \Cref{eq:iv3_2}.} Since $\Sigma_y$ is a p.d matrix, we can decompose it as $\sum_{i=1}^{n}\sigma_i\mu_i\mu_i\t$, where $\sigma_1\geq \sigma_2\geq\cdots\geq\sigma_{n}$ are the eigenvalues and $\{\mu_i\}_{i=1}^{n}$ are the corresponding eigenvectors that form an orthonormal basis of $\bbR^{n}$. Therefore, $\overline{\Sigma}_y = \Sigma_y+ TI = \Sigma_y + T\sum_{i=1}^{n}\mu_i\mu_i\t = \sum_{i=1}^{n} (\sigma_i+T)\mu_i\mu_i\t$. As a direct consequence, $\sigma_1\b{\Sigma_y^{-1}\overline{\Sigma}_y} = \max_i (\sigma_i+T)/\sigma_i = (\sigma_n+T)/\sigma_n$. 
    Substituting back gives
    \begin{equation}\begin{split}
        \norm{\b{\Sigma_y}^{-1} \sum_{t\in[T]} y_t\zeta_{t}\t} \leq {}& \norm{\Sigma_y^{-1}\overline{\Sigma}_y}\norm{\overline{\Sigma}_y^{-\frac{1}{2}}}\norm{\overline{\Sigma}_y^{-\frac{1}{2}} \sum_{t\in[T]} y_t\zeta_{t}\t}\\
        \lesssim {}& \frac{\sigma_n+T}{\sigma_n}\frac{1}{\sqrt{\sigma_n+T}}\sqrt{\max\{a,b\}\psi_\zeta\log\b{\frac{5(\sigma_1/T+1)}{\delta}}}.
    \end{split}\end{equation}
\end{proof}

\begin{corollary}\label{cor:cross_general}
    Consider the setting of \Cref{lem:iv3}.
    If $\{y_t\}_{t=1}^T$ is a trajectory from system $(A,B,\Sigma_w,\Sigma_u,\Sigma_{\eta})$ satisfying \Cref{assmp3:stable} with constants $\psi_A$ and $\rho_A$, then with probability $1-\delta$ for any $\delta\in(0,1)$
    \begin{equation}\begin{split}\label{eq:cross_general_1}
        \norm{\sum_{t=1}^{T} y_t\zeta_t\t} \lesssim \sqrt{\frac{\psi\psi_A^2\psi_\zeta}{1-\rho_A^2} n\max\{a,b\}T\log\frac{3}{\delta}\log\b{\frac{3\psi\psi_A^2n}{(1-\rho_A^2)\delta}\log\frac{3}{\delta}}}.
    \end{split}\end{equation}
    Here $\psi \coloneqq \psi_B^2\psi_u+\psi_w$, $\psi_B \coloneqq \max\{\norm{B},1\}$, $\psi_u \coloneqq \max\{\norm{\Sigma_u},1\}$, $\psi_w \coloneqq \max\{\norm{\Sigma_w},1\}$, $\psi_{\zeta} = \sigma_1\b{\Sigma_\zeta}$.

    If $\{y_t\}_{t=1}^T$ is a sequence of i.i.d. Gaussian noises from $\calN(0,\Sigma_y)$, then with probability at least $1-\delta$ for any $\delta\in(0,1)$
    \begin{equation}\begin{split}\label{eq:cross_general_3}
        \norm{\sum_{t=1}^{T} y_t\zeta_t\t} \lesssim \sqrt{\psi_y\psi_\zeta \max\{a,b\}T\log\b{\frac{3\psi_y}{\delta}}}.
    \end{split}\end{equation}
    Here $\psi_y = \max\{\norm{\Sigma_y}, 1\}$.
\end{corollary}

\begin{proof}
    
    \textbf{Now we prove \Cref{eq:cross_general_1}.} If $\{y_t\}_{t=1}^T$ is a trajectory from system $(A,B,\Sigma_w,\Sigma_u,\Sigma_{\eta})$, then with probability at least $1-\delta'$ for any $\delta'\in(0,e^{-1})$
    \begin{equation}\begin{split}
        \norm{ \sum_{t=0}^T y_ty_t\t} \lesssim \frac{\psi\psi_A^2n}{1-\rho_A^2} T \log\frac{1}{\delta'}.
    \end{split}\end{equation}
    Combining with \Cref{eq:iv3_1} ($\delta=2\delta'$), we get the following with probability at least $1-3\delta'$ for any $\delta'\in(0,e^{-1})$
    \begin{equation}\begin{split}
        \norm{\sum_{t\in[T]}y_t\zeta_t}
        \lesssim {}& \norm{\b{\overline{\Sigma}_y}^{\frac{1}{2}}}\sqrt{\max\{a,b\}\psi_\zeta\log\b{\frac{5(\sigma_1(\sum_{t=1}^{T}y_ty_t\t)/T+1)}{2\delta'}}}\\
        \lesssim {}& \sqrt{\frac{\psi\psi_A^2n}{1-\rho_A^2} T \log\frac{1}{\delta'}}\sqrt{\max\{a,b\}\psi_\zeta\log\b{5\frac{\frac{\psi\psi_A^2n}{1-\rho_A^2} \log\frac{1}{\delta'}+1}{2\delta'}}}\\
        \lesssim {}& \sqrt{\frac{\psi\psi_A^2\psi_\zeta}{1-\rho_A^2} n\max\{a,b\}T\log\frac{1}{\delta'}\log\b{\frac{\psi\psi_A^2n}{(1-\rho_A^2)\delta'}\log\frac{1}{\delta'}}}.
    \end{split}\end{equation}
    Finally, letting $\delta = 3\delta'$ gives the desired result.

    \textbf{Now we prove \Cref{eq:cross_general_3}.}  If $\{y_t\}_{t=1}^T$ is a sequence of i.i.d. Gaussian noises from $\calN(0,\Sigma_y)$, then by a standard covariance concentration result (Equation (45) of \cite{proof_2}), we have the following inequality for $T \gtrsim n + \log\b{\frac{1}{\delta'}}$ with probability at least $1-\delta'$
    \begin{equation}\begin{split}
        \frac{3}{4}T\Sigma_y \preceq \sum_{t=0}^{T-1} y_ty_t\t \preceq \frac{5}{4}T\Sigma_y.
    \end{split}\end{equation}
    This implies $\sigma_1(\sum_{t=1}^{T}y_ty_t\t) \leq 5T\sigma_1\b{\Sigma_y}/4$.
    Combining with \Cref{eq:iv3_1} ($\delta=2\delta'$), we get the following with probability at least $1-3\delta'$ for any $\delta'\in(0,e^{-1})$
    \begin{equation}\begin{split}
        \norm{\sum_{t\in[T]}y_t\zeta_t}
        \lesssim {}& \norm{\b{\overline{\Sigma}_y}^{\frac{1}{2}}}\sqrt{\max\{a,b\}\psi_\zeta\log\b{\frac{5(\sigma_1(\sum_{t=1}^{T}y_ty_t\t)/T+1)}{2\delta'}}}\\
        \lesssim {}& \sqrt{T\psi_y}\sqrt{\max\{a,b\}\psi_\zeta\log\b{5\frac{5\psi_y/4+1}{2\delta'}}}\\
        \lesssim {}& \sqrt{\psi_y\psi_\zeta \max\{a,b\}T\log\b{\frac{\psi_y}{\delta'}}}.
    \end{split}\end{equation}
    Finally, letting $\delta = 3\delta'$ gives the desired result.
\end{proof}

\begin{proof}[Proof of \Cref{lem:iv1}]
    Throughout this proof, we consider $\delta'\in(0,\frac{1}{3})$, which ensures $1\leq \log(k/\delta') \leq k\log(1/\delta')$ for all positive integer $k$.
    By definition, we have
\begin{equation}\begin{split}\label{eq:zi_1}
    \htz_{t} = {}& \begin{bmatrix}
        \htx_{t}\\
        u_{t}
    \end{bmatrix} = \begin{bmatrix}
        Ax_{t-1} + Bu_{t-1} + w_{t-1} + \eta_t\\
        u_t
    \end{bmatrix}\\
    = {}& \begin{bmatrix}
        A & 0\\
        0 & I
    \end{bmatrix}i_t + \begin{bmatrix}
        Bu_{t-1} + w_{t-1} + \eta_t\\
        0
    \end{bmatrix}.
\end{split}\end{equation}
Then,
\begin{equation}\begin{split}\label{eq:zi_2}
    \sum_{t=1}^{T-1}\htz_t \hti_t\t = {}& \sum_{t=1}^{T-1} \b{\begin{bmatrix}
        A & 0\\
        0 & I
    \end{bmatrix}i_t + \begin{bmatrix}
        Bu_{t-1} + w_{t-1} + \eta_t\\
        0
    \end{bmatrix}}\b{i_t+\begin{bmatrix}
        \eta_{t-1}\\
        0
    \end{bmatrix}}\t\\
    = {}& \begin{bmatrix}
        A & 0\\
        0 & I
    \end{bmatrix}\sum_{t=1}^{T-1} i_ti_t\t + \sum_{t=1}^{T-1} \begin{bmatrix}
        Bu_{t-1} + w_{t-1} + \eta_t\\
        0
    \end{bmatrix}\begin{bmatrix}
        x_{t-1}\\
        u_t
    \end{bmatrix}\t\\
    {}& + \begin{bmatrix}
        A & 0\\
        0 & I
    \end{bmatrix}\sum_{t=1}^{T-1}\begin{bmatrix}
        x_{t-1}\\
        u_t
    \end{bmatrix} \begin{bmatrix}
        \eta_{t-1}\\
        0
    \end{bmatrix}\t + \sum_{t=1}^{T-1} \begin{bmatrix}
        Bu_{t-1} + w_{t-1} + \eta_t\\
        0
    \end{bmatrix}\begin{bmatrix}
        \eta_{t-1}\\
        0
    \end{bmatrix}\t
\end{split}\end{equation}
Here the first term can be seen as the signal term, while the rest terms can be seen as noise terms. 

Now we bound the noise terms using \Cref{cor:cross_general}, an adapted version of the well-established self-normalizing bound (Theorem 1 in \cite{proof_3}). We first consider $\sum_{t=1}^{T-1} \eta_{t-1}\eta_{t}\t$, which is part of the third noise term. 
Applying \Cref{cor:cross_general}, or more specifically \Cref{eq:cross_general_3}, on sequences $\{\eta_{t-1}\}_{t=1}^{T-1}$ and $\{\eta_t\}_{t\in[T-1]}$ gives the following with probability at least $1-\delta'$
\begin{equation}\begin{split}\label{eq:cross1}
    \norm{\sum_{t=1}^{T-1} \eta_{t-1}\eta_t\t} \lesssim \sqrt{\psi_{\eta}^2nT\log\b{\frac{\psi_\eta n}{\delta'}}}.
\end{split}\end{equation}
Similarly, we have the following inequalities, each with probability at least $1-\delta'$,
\begin{equation}\begin{split}\label{eq:cross2}
    \norm{\sum_{t=1}^{T-1} \eta_{t-1}u_t\t}, \norm{\sum_{t=1}^{T-1} \eta_{t-1}u_{t-1}\t} \lesssim {}& \sqrt{\psi_{u}\psi_{\eta}nT\log\b{\frac{\psi_{\eta} n}{\delta'}}},\\
    \norm{\sum_{t=1}^{T-1} \eta_{t-1}w_{t-1}\t} \lesssim {}& \sqrt{\psi_w\psi_{\eta}nT\log\b{\frac{\psi_{\eta} n}{\delta'}}},\\
    \norm{\sum_{t=1}^{T-1} u_{t-1}u_t\t} \lesssim {}& \sqrt{\psi_{u}^2nT\log\b{\frac{\psi_u n}{\delta'}}}\\
    \norm{\sum_{t=1}^{T-1} u_{t}w_{t-1}\t} \lesssim {}& \sqrt{\psi_{u}\psi_wnT\log\b{\frac{\psi_u n}{\delta'}}}\\
\end{split}\end{equation}
Additionally, notice that $\{x_t\}_{t=0}^{T-2}$ is generated by the target system $\calM = (A,B,\Sigma_w, \Sigma_u,\Sigma_{\eta})$. We apply \Cref{cor:cross_general}, or more specifically \Cref{eq:cross_general_1}, on $\{x_t\}_{t=0}^{T-2}$ and $\{\eta_{t}\}_{t=0}^{T-2}$ gives the following, each with probability at least $1-\delta'$
\begin{equation}\begin{split}\label{eq:cross3}
    \norm{\sum_{t=0}^{T-2} x_{t}\eta_{t}\t }\lesssim {}& \sqrt{\frac{\psi\psi_A^2\psi_\eta}{1-\rho_A^2}n^2T\log\b{\frac{1}{\delta'}}}\sqrt{\log\b{\frac{\psi\psi_A^2}{(1-\rho_A^2)\delta'}n\log\frac{1}{\delta'}}}.
\end{split}\end{equation}
Similarly,
\begin{equation}\begin{split}\label{eq:cross4}
    \norm{\sum_{t=0}^{T-2} x_{t}\b{Bu_{t} + w_{t} + \eta_{t+1}}\t }\lesssim {}& \sqrt{\frac{\psi\psi_A^2\b{\psi_\eta+\psi}}{1-\rho_A^2}n^2T\log\b{\frac{1}{\delta'}}}\sqrt{\log\b{\frac{\psi\psi_A^2}{(1-\rho_A^2)\delta'}n\log\frac{1}{\delta'}}}.
\end{split}\end{equation}
By definition $\psi = \max\{\psi_B^2\psi_u+\psi_w, \psi_\eta\}$, all the above noise terms can be upper bounded by the following term up to some absolute constant
\begin{equation}\begin{split}
    \sqrt{\frac{\psi^2\psi_A^2}{1-\rho_A^2}(m+n)^2T\log\b{\frac{1}{\delta'}}}\sqrt{\log\b{\frac{\psi\psi_A^2}{(1-\rho_A^2)\delta'}(m+n)\log\frac{1}{\delta'}}}.
\end{split}\end{equation}
Combining \Cref{eq:cross1,eq:cross2,eq:cross3,eq:cross4}, we have the following with probability at least $1-8\delta'$,
\begin{equation}\begin{split}
    {}& \norm{\sum_{t=1}^{T-1} \begin{bmatrix}
        Bu_{t-1} + w_{t-1} + \eta_t\\
        0
    \end{bmatrix}\begin{bmatrix}
        x_{t-1}\\
        u_t
    \end{bmatrix}\t}, \norm{\sum_{t=1}^{T-1}\begin{bmatrix}
        x_{t-1}\\
        u_t
    \end{bmatrix} \begin{bmatrix}
        \eta_{t-1}\\
        0
    \end{bmatrix}\t}, \norm{\sum_{t=1}^{T-1} \begin{bmatrix}
        Bu_{t-1} + w_{t-1} + \eta_t\\
        0
    \end{bmatrix}\begin{bmatrix}
        \eta_{t-1}\\
        0
    \end{bmatrix}\t}\\
    \lesssim {}& \sqrt{\frac{\psi^2\psi_A^2}{1-\rho_A^2}(m+n)^2T\log\b{\frac{1}{\delta'}}}\sqrt{\log\b{\frac{\psi\psi_A^2}{(1-\rho_A^2)\delta'}(m+n)\log\frac{1}{\delta'}}}.
\end{split}\end{equation}

Now let $\Sigma_i \coloneqq \sum_{t=1}^{T-1}i_ti_t\t$ . Substituting back into \Cref{eq:zi_2} gives the following with probability at least $1-8\delta'$
\begin{equation}\begin{split}\label{eq:coiv_7}
    {}& \sigma_{m+n}\b{\b{\sum_{t=1}^{T-1}\htz_t \hti_t\t}\Sigma_i^{-1}}\\
    \geq {}& \sigma_{m+n}\b{\begin{bmatrix}
        A & 0\\
        0 & I
    \end{bmatrix}} - \sigma_1\b{\sum_{t=1}^{T-1} \begin{bmatrix}
        Bu_{t-1} + w_{t-1} + \eta_t\\
        0
    \end{bmatrix}\begin{bmatrix}
        x_{t-1}\\
        u_t
    \end{bmatrix}\t\Sigma_i^{-1}}\\
    {}& - \sigma_1\b{\sum_{t=1}^{T-1}\begin{bmatrix}
        x_{t-1}\\
        u_t
    \end{bmatrix} \begin{bmatrix}
        \eta_{t-1}\\
        0
    \end{bmatrix}\t\Sigma_i^{-1}} - \sigma_1\b{\sum_{t=1}^{T-1} \begin{bmatrix}
        Bu_{t-1} + w_{t-1} + \eta_t\\
        0
    \end{bmatrix}\begin{bmatrix}
        \eta_{t-1}\\
        0
    \end{bmatrix}\t\Sigma_i^{-1}}\\
    \geq {}& \min\{\phi_A,1\} - \norm{\sum_{t=1}^{T-1} \begin{bmatrix}
        Bu_{t-1} + w_{t-1} + \eta_t\\
        0
    \end{bmatrix}\begin{bmatrix}
        x_{t-1}\\
        u_t
    \end{bmatrix}\t}\norm{\Sigma_i^{-1}}\\
    {}& - \norm{\sum_{t=1}^{T-1}\begin{bmatrix}
        x_{t-1}\\
        u_t
    \end{bmatrix} \begin{bmatrix}
        \eta_{t-1}\\
        0
    \end{bmatrix}\t}\norm{\Sigma_i^{-1}} - \norm{\sum_{t=1}^{T-1} \begin{bmatrix}
        Bu_{t-1} + w_{t-1} + \eta_t\\
        0
    \end{bmatrix}\begin{bmatrix}
        \eta_{t-1}\\
        0
    \end{bmatrix}\t}\norm{\Sigma_i^{-1}}\\
\end{split}\end{equation}
By \Cref{lem:iv2}, $\sigma_{m+n}\b{\Sigma_i} \geq \min\{\phi_R^2/16,3/8\}\phi_uT$ with probability at least $1-\delta'$ under \Cref{eq:Tcond_iv} with $\delta=9\delta'$. Substituting back into \Cref{eq:coiv_7} gives the following with probability at least $1-9\delta'$
\begin{equation}\begin{split}\label{eq:coiv_8}
    {}& \sigma_{m+n}\b{\b{\sum_{t=1}^{T-1}\htz_t \hti_t\t}\Sigma_i^{-1}}\\
    \geq {}& \min\{\phi_A,1\}\\
    {}& - \dfrac{c_1}{\min\{\phi_R^2/16,3/8\}\phi_u}\sqrt{\frac{\psi^2\psi_A^2}{1-\rho_A^2}(m+n)^2\log\b{\frac{1}{\delta'}}}\sqrt{\log\b{\frac{\psi\psi_A^2}{(1-\rho_A^2)\delta'}(m+n)\log\frac{1}{\delta'}}}\sqrt{\frac{1}{T}}\\
    \overset{(i)}{\geq} {}& \frac{\min\{\phi_A,1\}}{2}.
\end{split}\end{equation}
Here $(i)$ is by \Cref{eq:Tcond_iv} with $\delta=9\delta'$. Finally, let $\delta = 9\delta'$, and we get the desired result.

\end{proof}

\subsection{Supporting Details}
The following lemmas hold for state trajectory $\{x_t\}_{t=0}^T$ from system $\calM = (A,B, \Sigma_w, \Sigma_u, \Sigma_{\eta})$.
\begin{lemma}[Lemma A.3 in \cite{proof_1}]\label{lem:cov}
    Suppose $A\in\bbR^{n\times n}, B\in\bbR^{n\times m}$ satisfy \Cref{assmp3:stable} with constants $\psi_A,\rho_A$ and \Cref{assmp2:controllable} with constant $\phi_R$. If $T \gtrsim \frac{\psi^2\psi_A^4}{\phi_u^2\phi_R^4}n^3\log\b{\frac{n}{\delta}}\log\b{\frac{\psi\psi_A^4}{1-\rho_A^2}n\log\frac{n}{\delta}}$, then the following holds with probability at least $1-\delta$ for any $\delta\in(0,1)$
    \begin{equation}\begin{split}
        \sum_{t=1}^{T} x_tx_t\t \succeq \frac{\phi_u\phi_R^2}{8}TI.
    \end{split}\end{equation}
    Here $\psi \coloneqq \psi_B^2\psi_u+\psi_w$, $\psi_B \coloneqq \max\{\norm{B},1\}$, $\psi_u \coloneqq \max\{\norm{\Sigma_u},1\}$, $\psi_w \coloneqq \max\{\norm{\Sigma_w},1\}$, $\phi_u \coloneqq \sigma_m\b{\Sigma_u}$. 
\end{lemma}

\begin{lemma}\label{lem:cov_upper}[Adapted from Corollary 8.2 in \cite{proof_2}]
    Suppose $A$ satisfies \Cref{assmp3:stable} with constant $\psi_A,\rho_A$. Then the following holds with probability at least $1-\delta$ for any $\delta\in(0,e^{-1})$
    \begin{equation}\begin{split}
        \norm{ \sum_{t=0}^T x_tx_t\t} \lesssim \frac{\psi\psi_A^2n}{1-\rho_A^2} T \log\frac{1}{\delta}.
    \end{split}\end{equation}
    Here $\psi \coloneqq \psi_B^2\psi_u+\psi_w$, $\psi_B \coloneqq \max\{\norm{B},1\}$, $\psi_u \coloneqq \max\{\norm{\Sigma_u},1\}$, $\psi_w \coloneqq \max\{\norm{\Sigma_w},1\}$. 
\end{lemma}

\begin{proof}
    Directly applying Proposition 8.4 in \cite{proof_2} gives the following with probability at least $1-\delta$
    \begin{equation}\begin{split}
        \norm{\sum_{t=0}^T x_tx_t\t} \lesssim {}& \psi\tr\b{\sum_{t=0}^{T-1} \Gamma_t(A)}\log\frac{1}{\delta} \leq \psi n \norm{\sum_{t=0}^{T-1} \Gamma_t(A)} \log\frac{1}{\delta}.
    \end{split}\end{equation}
    From Assumption \ref{assmp3:stable}, we know that 
    \begin{equation}\begin{split}
        \norm{\Gamma_t(A)} = \norm{\sum_{\tau=0}^t A^{\tau}A^{\tau}{}\t} \preceq \sum_{\tau=0}^t \norm{A^\tau}^2 \leq 1 + \frac{\psi_A^2}{1-\rho_A^2} \lesssim \frac{\psi_A^2}{1-\rho_A^2}.
    \end{split}\end{equation}
    Substituting back gives
    \begin{equation}\begin{split}
        \norm{\sum_{t=0}^T x_tx_t\t} \lesssim \frac{\psi\psi_A^2n}{1-\rho_A^2}T\log\frac{1}{\delta}.
    \end{split}\end{equation}
\end{proof}

\section{Proof of \Cref{thm:bc}}\label{sc:proof_bc}
\begin{proof}
    For simplicity, define $\tilde{\eta}_t \coloneqq \begin{bmatrix}
        \eta_t\\
        0
    \end{bmatrix}$, $\Sigma_{\tilde{\eta}} \coloneqq \begin{bmatrix}
        \Sigma_\eta & 0\\
        0 & 0
    \end{bmatrix}$, $\wh{\Sigma}_{\tilde{\eta}} \coloneqq \begin{bmatrix}
        \wh{\Sigma}_{\eta} & 0\\
        0 & 0
    \end{bmatrix}$, $\Sigma_{\htz} = \sum_{t=0}^{T-1}\htz_t\htz_t\t$ and $\Sigma_{z} = \sum_{t=0}^{T-1}z_tz_t\t$.
    Recall that by definition (\Cref{eq:estimator_naive2}), we have
    \begin{equation}\begin{split}\label{eq:bcmain_8}
        \htE_{\ls} = {}& E + \underbrace{\sum_{t=0}^{T-1} \b{- A\eta_tz_t\t + \eta_{t+1}\htz_t\t+w_t\htz_t\t}\b{\Sigma_{\htz}}^{-1}}_{\Delta_1} - \sum_{t=0}^{T-1}{E\tilde{\eta}_t\tilde{\eta}_t\t\b{\Sigma_{\htz}}^{-1}}\\
        = {}& E\b{I - \b{\sum_{t=0}^{T-1}\tilde{\eta}_t\tilde{\eta}_t\t}\b{\Sigma_{\htz}}^{-1}} + \Delta_1.
    \end{split}\end{equation}
    For the rest of the proof, we will first verify the invertibility of matrix $I - T\wh{\Sigma}_{\tilde{\eta}}\b{\Sigma_{\htz}}^{-1}$ (\textit{Step 1}). We then decompose and bound the error $\htE_{\bc}-E$ (\textit{Step 2}).

    \textit{Step 1: Invertibility of $\mathit{I - T\wh{\Sigma}_{\tilde{\eta}}\b{\Sigma_{\htz}}^{-1}}$.}
    By \Cref{lem:bc_covnoise}, with probability at least $1-\delta$
    \begin{equation}\begin{split}\label{eq:bcmain_2}
        \min\left\{\frac{\phi_R^2}{16}, \frac{3}{8}\right\}\phi_uTI \preceq \Sigma_{\htz} \precsim \frac{\psi\psi_A^2}{1-\rho_A^2}\log\b{\frac{4}{\delta}}nTI, \quad \sigma_{m+n}\b{\Sigma_{\htz}\Sigma_z^{-1}} \geq \frac{1}{3}.
    \end{split}\end{equation}
    This gives the following with probability at least $1-\delta$
    \begin{equation}\begin{split}\label{eq:bcmain_3}
        \sigma_{m+n}\b{I - T\wh{\Sigma}_{\tilde{\eta}}\b{\Sigma_{\htz}}^{-1}}
        \geq {}& \sigma_{m+n}\b{I} - \sigma_1\b{T\wh{\Sigma}_{\tilde{\eta}}\b{\Sigma_{\htz}}^{-1}}\\
        \geq {}& \sigma_{m+n}\b{I} - \sigma_1\b{T\wh{\Sigma}_{\tilde{\eta}}}\sigma_1\b{\b{\Sigma_{\htz}}^{-1}}\\
        \geq {}& 1 - \frac{\b{\psi_\eta+\epsilon_\eta}T}{\min\{\phi_R^2/16, 3/8\}\phi_uT}\\
        \geq {}& \frac{1}{4}.
    \end{split}\end{equation}
    Here the last inequality is by \Cref{assmp:input}.

    \textit{Step 2: Error Upperbound on $\mathit{\norm{\htE_{\bc}-E}}$.}
    By definition, we have
    \begin{equation}\begin{split}\label{eq:errdecomp_bc}
        {}& \htE_{\bc} - E = \htE_{\ls}\b{I - T\wh{\Sigma}_{\tilde{\eta}}\b{\Sigma_{\htz}}^{-1}}^{-1} - E\\
        = {}& E\b{I - \b{\sum_{t=0}^{T-1}\tilde{\eta}_t\tilde{\eta}_t\t}\b{\Sigma_{\htz}}^{-1}}\b{I - T\wh{\Sigma}_{\tilde{\eta}}\b{\Sigma_{\htz}}^{-1}}^{-1} - E + \Delta_1\b{I - T\wh{\Sigma}_{\tilde{\eta}}\b{\Sigma_{\htz}}^{-1}}^{-1}\\
        = {}& E\b{T\wh{\Sigma}_{\tilde{\eta}}-\sum_{t=0}^{T-1}\tilde{\eta}_t\tilde{\eta}_t\t}\b{\Sigma_{\htz}}^{-1}\b{I - T\wh{\Sigma}_{\tilde{\eta}}\b{\Sigma_{\htz}}^{-1}}^{-1} + \Delta_1\b{I - T\wh{\Sigma}_{\tilde{\eta}}\b{\Sigma_{\htz}}^{-1}}^{-1}\\
    \end{split}\end{equation}

    Consider the first term in the last line. 
    By a standard convariance concentration result (Lemma A.5 in \cite{proof_1}), we get $\norm{\sum_{t=0}^{T-1} -\tilde{\eta}_t\tilde{\eta}_t\t + T\Sigma_{\tilde{\eta}}} = \norm{\sum_{t=0}^{T-1} -\eta_t\eta_t\t + T\Sigma_{\eta}} \lesssim \psi_{\eta}\sqrt{nT\log\frac{1}{\delta}}$ with probability at least $1-\delta$ for $T$ satisfying \Cref{eq:Tcond_bc}.
    Therefore, with probability at least $1-\delta$
    \begin{equation}\begin{split}\label{eq:bcmain_4}
        \norm{T\wh{\Sigma}_{\tilde{\eta}}-\sum_{t=0}^{T-1}\tilde{\eta}_t\tilde{\eta}_t\t} \leq T\norm{\wh{\Sigma}_{\tilde{\eta}} - \Sigma_{\tilde{\eta}}} + \norm{T\Sigma_{\tilde{\eta}} - \sum_{t=0}^{T-1}\tilde{\eta}_t\tilde{\eta}_t\t} \lesssim T\epsilon_{\eta} + \psi_{\eta}\sqrt{nT\log\frac{1}{\delta}}.
    \end{split}\end{equation}
    Substituting \Cref{eq:bcmain_2,eq:bcmain_3,eq:bcmain_4} into \Cref{eq:errdecomp_bc} gives the following with probability at least $1-3\delta$
    \begin{equation}\begin{split}\label{eq:bcmain_7}
        {}& \norm{E\b{T\wh{\Sigma}_{\tilde{\eta}}-\sum_{t=0}^{T-1}\tilde{\eta}_t\tilde{\eta}_t\t}\b{\Sigma_{\htz}}^{-1}\b{I - T\wh{\Sigma}_{\tilde{\eta}}\b{\Sigma_{\htz}}^{-1}}^{-1}}\\
        = {}& \norm{E\b{T\wh{\Sigma}_{\tilde{\eta}}-\sum_{t=0}^{T-1}\tilde{\eta}_t\tilde{\eta}_t\t}}\norm{\b{\Sigma_{\htz}}^{-1}}\norm{\b{I - T\wh{\Sigma}_{\tilde{\eta}}\b{\Sigma_{\htz}}^{-1}}^{-1}}\\
        \lesssim {}& \frac{\psi_{\eta}}{\min\{\phi_R^2,6\}\phi_u}\sqrt{\frac{n}{T}\log\frac{1}{\delta}}+\frac{1}{\min\{\phi_R^2,6\}\phi_u}\epsilon_{\eta}.
    \end{split}\end{equation}

    Now consider term $\sum_{t=0}^{T-1} \b{\eta_{t+1}+w_t}\htz_t\t\Sigma_{\htz}^{-1}$ in term $\Delta_1$ (\Cref{eq:errdecomp_bc}). Notice that $\eta_{t+1}+w_t$ is i.i.d. for $t\in[0,T-1]$ and independent of $\htz_{0:t}$, i.e. independent of $u_{0:t}$, $x_{0:t}$ and $\eta_{0:t}$. Additionally, by \Cref{lem:bc_covnoise}, with probability at least $1-\delta$,
    \begin{equation}\begin{split}\label{eq:bcmain_9}
        \min\left\{\frac{\phi_R^2}{16}, \frac{3}{8}\right\}\phi_uTI \preceq \Sigma_{\htz} \precsim \frac{\psi\psi_A^2}{1-\rho_A^2}\log\b{\frac{4}{\delta}}nTI, \quad \sigma_{m+n}\b{\Sigma_{\htz}\Sigma_z^{-1}} \geq \frac{1}{3}.
    \end{split}\end{equation}
    We can then apply \Cref{lem:iv3} (\Cref{eq:iv3_2}) and get the following with probability at least $1-2\delta$
    \begin{equation}\begin{split}\label{eq:bcmain_5}
        \norm{\b{\Sigma_{\htz}}^{-1} \sum_{t=0}^{T-1}\htz_t\b{\eta_{t+1}+w_t}\t} \lesssim \frac{\sqrt{\min\left\{\phi_R^2, 6\right\}\phi_u+16}}{\min\left\{\phi_R^2, 6\right\}\phi_u}\sqrt{\frac{m+n}{T}\b{\psi_\eta+\psi_w}\log\b{\frac{5\psi\psi_A^2}{(1-\rho_A^2)\delta}n\log\frac{4}{\delta}}}.
    \end{split}\end{equation}
    For term $\sum_{t=0}^{T-1} A\eta_tz_t\t\b{\sum_{t=0}^{T-1} \htz_t\htz_t\t}^{-1}$ in $\Delta_1$ (\Cref{eq:errdecomp_bc}). Notice that $\eta_{t}$ is i.i.d. for $t\in[0,T-1]$ and independent of $z_t$, i.e. independent of $u_{0:t}$, $x_{0:t}$. Additionally, by \Cref{lem:bc_cov}, with probability at least $1-\delta$,
    \begin{equation}\begin{split}
        \min\left\{\frac{\phi_R^2}{12}, \frac{1}{2}\right\}\phi_uTI \preceq \Sigma_z \precsim \frac{\psi\psi_A^2}{1-\rho_A^2}\log\b{\frac{4}{\delta}}nTI.
    \end{split}\end{equation}    
    We can then apply \Cref{lem:iv3} (\Cref{eq:iv3_2}), an adapted version of the self-normalizing bound (Theorem 1 in \cite{proof_3}) and get the following with probability at least $1-2\delta$
    \begin{equation}\begin{split}\label{eq:bc_4}
        \norm{\b{\Sigma_{z}}^{-1} \sum_{t=0}^{T-1}z_t\eta_{t}\t} \lesssim \frac{\sqrt{\min\left\{\phi_R^2, 6\right\}\phi_u+12}}{\min\left\{\phi_R^2, 6\right\}\phi_u}\sqrt{\frac{m+n}{T}\psi_\eta\log\b{\frac{5\psi\psi_A^2}{(1-\rho_A^2)\delta}n\log\frac{4}{\delta}}}.
    \end{split}\end{equation}
    Therefore, 
    combining \Cref{eq:bcmain_9,eq:bc_4}, we have  the following with probability at least $1-3\delta$
    \begin{equation}\begin{split}\label{eq:bcmain_6}
        {}& \norm{\b{\Sigma_{\htz}}^{-1}\sum_{t=0}^{T-1} z_t\eta_t\t A\t}\leq \norm{\b{\Sigma_{z}}^{-1}\sum_{t=0}^{T-1} z_t\eta_t\t A\t}\norm{\Sigma_{z}\b{\Sigma_{\htz}}^{-1}}\\
        \lesssim {}& \psi_A\frac{\sqrt{\min\left\{\phi_R^2, 6\right\}\phi_u+12}}{\min\left\{\phi_R^2, 6\right\}\phi_u}\sqrt{\frac{m+n}{T}\psi_\eta\log\b{\frac{5\psi\psi_A^2}{(1-\rho_A^2)\delta}n\log\frac{4}{\delta}}}.
    \end{split}\end{equation}
    Finally, combining \Cref{eq:bcmain_6,eq:bcmain_5,eq:bcmain_3} gives the following with probability at least $1-6\delta$
    \begin{equation}\begin{split}\label{eq:bcmain_10}
        {}& \norm{\Delta_1\b{I - T\wh{\Sigma}_{\tilde{\eta}}\b{\Sigma_{\htz}}^{-1}}^{-1}} \leq \norm{\Delta_1}\norm{\b{I - T\wh{\Sigma}_{\tilde{\eta}}\b{\Sigma_{\htz}}^{-1}}^{-1}}\\
        \lesssim {}& \psi_A\frac{\sqrt{\min\left\{\phi_R^2, 1\right\}\phi_u+1}}{\min\left\{\phi_R^2, 1\right\}\phi_u}\sqrt{\frac{m+n}{T}\psi_\eta\log\b{\frac{5\psi\psi_A^2}{(1-\rho_A^2)\delta}n\log\frac{4}{\delta}}}
    \end{split}\end{equation}

    Finally, we substitute \Cref{eq:bcmain_7,eq:bcmain_10} into \Cref{eq:errdecomp_bc} and get the following with probability at least $1-9\delta$
    \begin{equation*}\begin{split}
        \norm{\htE_{\bc} - E} 
        \lesssim {}& \frac{\psi_{\eta}}{\min\{\phi_R^2,6\}\phi_u}\sqrt{\frac{n}{T}\log\frac{1}{\delta}}+\frac{1}{\min\{\phi_R^2,6\}\phi_u}\epsilon_{\eta}\\
        {}& + \psi_A\frac{\sqrt{\min\left\{\phi_R^2, 1\right\}\phi_u+1}}{\min\left\{\phi_R^2, 1\right\}\phi_u}\sqrt{\frac{m+n}{T}\psi_\eta\log\b{\frac{5\psi\psi_A^2}{(1-\rho_A^2)\delta}n\log\frac{4}{\delta}}}\\
        \lesssim {}& \frac{1}{\min\{\phi_R^2,1\}\phi_u}\epsilon_{\eta} + \psi_A\sqrt{\frac{\min\left\{\phi_R^2, 1\right\}\phi_u+1}{\min\left\{\phi_R^6, 1\right\}\phi_u}}\max\left\{\sqrt{\frac{\psi}{\phi_u}},\frac{\psi}{\phi_u}\right\}\cdot\sqrt{\frac{m+n}{T}\log\b{\frac{5\psi\psi_A^2}{(1-\rho_A^2)\delta}n\log\frac{4}{\delta}}}.
    \end{split}\end{equation*}

\end{proof}

\subsection{Supporting Details}
\begin{lemma}\label{lem:bc_cov}
    Consider the setting of \Cref{thm:bc}.
    Suppose $T$ satisfies \Cref{eq:Tcond_bc}.
    Then with probability at least $1-\delta$ for any $\delta\in(0,1)$,
    we know that
    \begin{equation}\begin{split}
        \min\left\{\frac{\phi_R^2}{12}, \frac{1}{2}\right\}\phi_uTI \preceq \sum_{t=0}^{T-1}z_tz_t\t \precsim \frac{\psi\psi_A^2}{1-\rho_A^2}\log\b{\frac{4}{\delta}}nTI.
    \end{split}\end{equation}
\end{lemma}

\begin{proof}
    Throughout this proof, we consider $\delta'\in(0,\frac{1}{3})$, which ensures $1\leq \log(k/\delta') \leq k\log(1/\delta')$ for all positive integer $k$.
    By definition, we have
    \begin{equation}\begin{split}\label{eq:cobc_1}
        \sum_{t=0}^{T-1} z_tz_t\t = {}& \sum_{t=0}^{T-1} \begin{bmatrix}
            x_t\\
            u_t
        \end{bmatrix}\begin{bmatrix}
            x_t\\
            u_t
        \end{bmatrix}\t= \sum_{t=0}^{T-1} \begin{bmatrix}
            x_tx_t\t & 0\\
            0 & u_tu_t\t
        \end{bmatrix} + \begin{bmatrix}
            0 & x_{t}u_t\t\\
            u_tx_{t}\t & 0
        \end{bmatrix}.
    \end{split}\end{equation}

    We first show the signal terms, i.e. $\sum_{t=0}^{T-1} x_{t}x_{t}\t$ and $\sum_{t=0}^{T-1} u_tu_t\t$ are large. By a standard covariance concentration result (Equation (45) of \cite{proof_2}), we have the following inequalities for $T \gtrsim n + \log\b{\frac{1}{\delta'}}$, each with probability at least $1-\delta'$
\begin{equation}\begin{split}\label{eq:cobc_5}
    \frac{3}{4}T\Sigma_u \preceq \sum_{t=0}^{T-1} u_tu_t\t \preceq \frac{5}{4}T\Sigma_u.
\end{split}\end{equation}
Moreover, by a standard lower bound (\Cref{lem:cov}) and upper bound (\Cref{lem:cov_upper}) on the empirical state covariance, with probability at least $1-2\delta'$
\begin{equation}\begin{split}\label{eq:cobc_6}
    \frac{\phi_u\phi_R^2}{8}(T-1)I \preceq \sum_{t=0}^{T-1} x_{t}x_{t}\t \precsim \frac{\psi\psi_A^2n}{1-\rho_A^2}TI\log\frac{1}{\delta'}.
\end{split}\end{equation}
for any $T$ satisfies
\begin{equation}\begin{split}
    T \gtrsim\frac{\psi^2\psi_A^4}{\phi_u^2\phi_R^4}n^3\log\b{\frac{n}{\delta'}}\log\b{\frac{\psi\psi_A^4}{1-\rho_A^2}n\log\frac{n}{\delta'}}.
\end{split}\end{equation}

Now we prove that the cross terms are small. We note that $\{x_t\}_{t=1}^{T-1}$ is a trajectory from the target system $\calM = (A,B,\Sigma_w, \Sigma_u, \Sigma_\eta)$, and that $\{u_{t}\}_{t=1}^{T-1}$ a sequence of i.i.d. gaussian vectors such that $u_t$ is independent of $x_{0:t}$. Therefore, we can apply \Cref{cor:cross_general}, an adapted version of the well-established self-normalizing bound (Theorem 1 in \cite{proof_3}), on the two sequences and get the following with probability at least $1-\delta'$
\begin{equation}\begin{split}\label{eq:cobc_3}
    \norm{\sum_{t=0}^{T-1} x_{t}u_{t}\t } \lesssim {}& \sqrt{\frac{\psi\psi_A^2\psi_u}{1-\rho_A^2}n\max\{m,n\}T\log\b{\frac{1}{\delta'}}}\sqrt{\log\b{\frac{\psi\psi_A^2}{(1-\rho_A^2)\delta'}n\log\frac{1}{\delta'}}}.
\end{split}\end{equation}

Substituting \Cref{eq:cobc_3,eq:cobc_5,eq:cobc_6} back into \Cref{eq:cobc_1} gives the following for some absolute constant $c_1$ with probability at least $1-4\delta'$
\begin{equation}\begin{split}\label{eq:resbc}
    \sum_{t=0}^{T-1} z_tz_t\t \succeq {}& \sum_{t=0}^{T-1} \begin{bmatrix}
        x_{t}x_{t}\t & 0\\
        0 & u_tu_t\t
    \end{bmatrix} - \norm{\sum_{t=0}^{T-1}\begin{bmatrix}
        0 & x_{t}u_t\t\\
        u_tx_{t}\t & 0
    \end{bmatrix}}I\\
    \succeq {}& \min\left\{\frac{\phi_R^2}{8}, \frac{3}{4}\right\}\phi_u(T-1)I - c_1\sqrt{\frac{\psi^2\psi_A^2}{1-\rho_A^2}(m+n)^2\log\b{\frac{1}{\delta'}}}\sqrt{\log\b{\frac{\psi\psi_A^2}{(1-\rho_A^2)\delta'}n\log\frac{1}{\delta'}}}\sqrt{T}\\
    \succeq {}& \min\left\{\frac{\phi_R^2}{12}, \frac{1}{2}\right\}\phi_uTI.
\end{split}\end{equation}
Here the last line is by \Cref{eq:Tcond_bc} where $\delta=4\delta'$. Moreover,
\begin{equation}\begin{split}
    \sum_{t=0}^{T-1} z_tz_t\t \preceq {}& \norm{\sum_{t=0}^{T-1} \begin{bmatrix}
        x_{t}x_{t}\t & 0\\
        0 & u_tu_t\t
    \end{bmatrix}}I + \norm{\sum_{t=0}^{T-1}\begin{bmatrix}
        0 & x_{t}u_t\t\\
        u_tx_{t}\t & 0
    \end{bmatrix}}I\\
    \precsim {}& \max\left\{\frac{\psi\psi_A^2n}{1-\rho_A^2}T\log\frac{1}{\delta'}, \frac{5}{4}\psi_{u}T\right\}I + c_1\sqrt{\frac{\psi^2\psi_A^2}{1-\rho_A^2}(m+n)^2\log\b{\frac{1}{\delta'}}}\sqrt{\log\b{\frac{\psi\psi_A^2}{(1-\rho_A^2)\delta'}n\log\frac{1}{\delta'}}}\sqrt{T}\\
    \precsim {}& \frac{\psi\psi_A^2n}{1-\rho_A^2}TI\log\frac{1}{\delta'}.
\end{split}\end{equation}
Here the last line is by \Cref{eq:Tcond_bc} with $\delta=4\delta'$. 
Finally, let $\delta = 4\delta'$, and we get the desired result.

\end{proof}

\begin{lemma}\label{lem:bc_covnoise}
    Consider the setting of \Cref{thm:bc}.
    Suppose $T$ satisfies \Cref{eq:Tcond_bc}.
    Then with probability at least $1-\delta$ for any $\delta\in(0,1)$,
    we know that
    \begin{equation}\begin{split}\label{eq:covnoise_1}
        \min\left\{\frac{\phi_R^2}{16}, \frac{3}{8}\right\}\phi_uTI \preceq \sum_{t=0}^{T-1}\htz_t\htz_t\t \precsim \frac{\psi\psi_A^2}{1-\rho_A^2}\log\b{\frac{3}{\delta}}nTI.
    \end{split}\end{equation}
    and that 
    \begin{equation}\begin{split}\label{eq:covnoise_2}
        \sigma_{m+n}\b{\b{\sum_{t=0}^{T-1} \htz_t\htz_t\t}\Sigma_z^{-1}} \geq \frac{1}{3}.
    \end{split}\end{equation}
\end{lemma}

\begin{proof}
    Throughout this proof, we consider $\delta'\in(0,\frac{1}{3})$, which ensures $1\leq \log(k/\delta') \leq k\log(1/\delta')$ for all positive integer $k$.
    By definition, we have that
    \begin{equation}\begin{split}\label{eq:cov_bc}
        \sum_{t=0}^{T-1} \htz_t\htz_t\t = {}& \sum_{t=0}^{T-1} z_tz_t\t + z_t\begin{bmatrix}
            \eta_t\\
            0
        \end{bmatrix}\t + \begin{bmatrix}
            \eta_t\\
            0
        \end{bmatrix}z_t\t + \begin{bmatrix}
            \eta_t\\
            0
        \end{bmatrix}\begin{bmatrix}
            \eta_t\\
            0
        \end{bmatrix}\t.
    \end{split}\end{equation}

    For the first covariance term, we apply \Cref{lem:bc_cov} and get the following with probability at least $1-\delta'$,
    \begin{equation}\begin{split}\label{eq:covbc_1}
        \min\left\{\frac{\phi_R^2}{12}, \frac{1}{2}\right\}\phi_uTI \preceq \sum_{t=0}^{T-1}z_tz_t\t \precsim \frac{\psi\psi_A^2}{1-\rho_A^2}\log\b{\frac{1}{\delta'}}nTI.
    \end{split}\end{equation}
    For the last noise covariance, by a standard covariance concentration result (Equation (45) of \cite{proof_2}), we have the following for $T \gtrsim n + \log\b{\frac{2}{\delta'}}$ with probability at least $1-\delta'$
\begin{equation}\begin{split}\label{eq:crossbc_3}
    \frac{3}{4}T\Sigma_{\eta} \preceq \sum_{t=0}^{T-1} \eta_t\eta_t\t \preceq \frac{5}{4}T\Sigma_\eta.
\end{split}\end{equation}
for any $T$ satisfies $T \gtrsim\frac{\psi^2\psi_A^4}{\phi_u^2\phi_R^4}n^3\log\b{\frac{n}{\delta'}}\log\b{\frac{\psi\psi_A^4}{1-\rho_A^2}n\log\frac{n}{\delta'}}$.

Now we bound the noise terms using \Cref{cor:cross_general}, an adapted version of the well-established self-normalizing bound (Theorem 1 in \cite{proof_3}). We note that $\eta_t$ is independent of $z_{0:t-1}$, i.e. $x_{0:t-1}$ and $u_{0:t-1}$. Therefore, we apply \Cref{eq:cross_general_1} on $\{z_t\}_{t=0}^{T-1}$ and $\{\eta_{t}\}_{t=0}^{T-1}$ and get the following with probability at least $1-\delta'$
\begin{equation}\begin{split}\label{eq:crossbc_2}
    \norm{\sum_{t\in[T]} z_t\eta_t\t }\lesssim {}& \sqrt{\frac{\psi^2\psi_A^2}{1-\rho_A^2}(m+n)^2T\log\b{\frac{1}{\delta'}}\log\b{\frac{\psi\psi_A^2n}{\b{1-\rho_A^2}\delta'}\log\frac{1}{\delta'}}}.
\end{split}\end{equation}

\textbf{We first use the above inequalities to prove \Cref{eq:covnoise_1}.} Substituting back into \Cref{eq:cov_bc} gives the following for some absolute constant $c_1$ with probability at least $1-3\delta'$
\begin{equation}\begin{split}\label{eq:crossbc_4}
    \sum_{t=0}^{T-1} \htz_t\htz_t\t \succeq {}& \sum_{t=0}^{T-1} z_tz_t\t - 2\norm{\sum_{t=0}^{T-1} z_t\begin{bmatrix}
        \eta_t\\
        0
    \end{bmatrix}\t}I\\
    \succeq {}& \min\left\{\frac{\phi_R^2}{12}, \frac{1}{2}\right\}\phi_uTI - c_1\sqrt{\frac{\psi^2\psi_A^2}{1-\rho_A^2}(m+n)^2T\log\b{\frac{1}{\delta'}}\log\b{\frac{\psi\psi_A^2n}{\b{1-\rho_A^2}\delta'}\log\frac{1}{\delta'}}}\sqrt{T}I\\
    \succeq {}& \min\left\{\frac{\phi_R^2}{16}, \frac{3}{8}\right\}\phi_uTI.
\end{split}\end{equation}
Here the last inequality is by \Cref{eq:Tcond_bc} with $\delta=3\delta'$. Moreover, 
\begin{equation}\begin{split}\label{eq:bcres_1}
    \sum_{t=0}^{T-1} \htz_t\htz_t\t \preceq {}& \norm{\sum_{t=0}^{T-1} z_tz_t\t}I + 2\norm{\sum_{t=0}^{T-1} z_t\begin{bmatrix}
        \eta_t\\
        0
    \end{bmatrix}\t}I + \norm{\begin{bmatrix}
        \eta_t\\
        0
    \end{bmatrix}\begin{bmatrix}
        \eta_t\\
        0
    \end{bmatrix}\t}\\
    \precsim {}& \frac{\psi\psi_A^2}{1-\rho_A^2}\log\b{\frac{1}{\delta'}}nTI + c_1\sqrt{\frac{\psi^2\psi_A^2}{1-\rho_A^2}(m+n)^2T\log\b{\frac{1}{\delta'}}\log\b{\frac{\psi\psi_A^2n}{\b{1-\rho_A^2}\delta'}\log\frac{1}{\delta'}}}\sqrt{T}I + \psi_{\eta}TI\\
    \precsim {}& \frac{\psi\psi_A^2}{1-\rho_A^2}\log\b{\frac{1}{\delta'}}nTI.
\end{split}\end{equation}
Substituting $\delta = 3\delta'$ gives the desired result.

\textbf{We now prove \Cref{eq:covnoise_2}. } Let $\Sigma_z \coloneqq \sum_{t=0}^{T-1} z_tz_t\t$. By \Cref{eq:cov_bc}, 
\begin{equation}\begin{split}\label{eq:crossbc_5}
    \b{\sum_{t=0}^{T-1} \htz_t\htz_t\t}\Sigma_z^{-1} = {}& I + \b{\sum_{t=0}^{T-1} z_t\begin{bmatrix}
        \eta_t\\
        0
    \end{bmatrix}\t + \begin{bmatrix}
        \eta_t\\
        0
    \end{bmatrix}z_t\t}\Sigma_z^{-1} + \b{\sum_{t=0}^{T-1}\begin{bmatrix}
        \eta_t\\
        0
    \end{bmatrix}\begin{bmatrix}
        \eta_t\\
        0
    \end{bmatrix}\t}\Sigma_z^{-1}.    
\end{split}\end{equation}
By \Cref{eq:crossbc_2,eq:covbc_1}, 
\begin{equation}\begin{split}
    {}& \norm{\b{\sum_{t=0}^{T-1} z_t\begin{bmatrix}
        \eta_t\\
        0
    \end{bmatrix}\t + \begin{bmatrix}
        \eta_t\\
        0
    \end{bmatrix}z_t\t}\Sigma_z^{-1}} \leq 2\norm{\sum_{t=0}^{T-1} z_t\eta_t\t}\norm{\Sigma_z^{-1}}\\
    \lesssim {}& \sqrt{\frac{\psi^2\psi_A^2}{(1-\rho_A^2)\phi_u^2\min\left\{\phi_R^4, 1\right\}}\log\b{\frac{1}{\delta'}}\log\b{\frac{\psi\psi_A^2n}{\b{1-\rho_A^2}\delta'}\log\frac{1}{\delta'}}} \sqrt{\frac{(m+n)^2}{T}}.
\end{split}\end{equation}
By \Cref{eq:crossbc_3,eq:covbc_1}, 
\begin{equation}\begin{split}
    \norm{\b{\sum_{t=0}^{T-1}\begin{bmatrix}
        \eta_t\\
        0
    \end{bmatrix}\begin{bmatrix}
        \eta_t\\
        0
    \end{bmatrix}\t}\Sigma_z^{-1}} \leq \norm{\sum_{t=0}^{T-1}\eta_t\eta_t\t}\norm{\Sigma_z^{-1}} \leq \frac{15\psi_{\eta}}{\min\left\{\phi_R^2, 6\right\}\phi_u}
\end{split}\end{equation}
Substituting back into \Cref{eq:crossbc_5} gives
\begin{equation}\begin{split}
    {}& \sigma_{m+n}\b{\b{\sum_{t=0}^{T-1} \htz_t\htz_t\t}\overline{\Sigma}_z^{-1}}\\
    \geq {}& 1 - \sigma_1\b{\b{\sum_{t=0}^{T-1} z_t\begin{bmatrix}
        \eta_t\\
        0
    \end{bmatrix}\t + \begin{bmatrix}
        \eta_t\\
        0
    \end{bmatrix}z_t\t}\overline{\Sigma}_z^{-1}} - \sigma_1\b{\b{\sum_{t=0}^{T-1}\begin{bmatrix}
        \eta_t\\
        0
    \end{bmatrix}\begin{bmatrix}
        \eta_t\\
        0
    \end{bmatrix}\t}\overline{\Sigma}_z^{-1}}\\
    \geq {}& 1 - \frac{15\psi_{\eta}}{\min\left\{\phi_R^2, 6\right\}\phi_u}\\
    {}&  - c_1\sqrt{\frac{\psi^2\psi_A^2}{(1-\rho_A^2)\phi_u^2\min\left\{\phi_R^4, 1\right\}}\log\b{\frac{1}{\delta'}}\log\b{\frac{\psi\psi_A^2n}{\b{1-\rho_A^2}\delta'}\log\frac{1}{\delta'}}} \sqrt{\frac{(m+n)^2}{T}}\\
\end{split}\end{equation}
Since $\min\{\phi_R^2,6\}\phi_u \geq 30\psi_{\eta}$ by \Cref{assmp:input} and $T$ satisfies \Cref{eq:Tcond_bc} ($\delta=3\delta'$), we have
\begin{equation}\begin{split}
    \sigma_{m+n}\b{\b{\sum_{t=0}^{T-1} \htz_t\htz_t\t}\Sigma_z^{-1}} \geq \frac{1}{3}.
\end{split}\end{equation}

\end{proof}

\end{document}